%% file: main.tex
\documentclass[aps,pra,10pt,notitlepage,twocolumn,nofootinbib,superscriptaddress]{revtex4-2}
\input{pretex}

\usepackage{qcircuit}
\usepackage{tikz}
\usetikzlibrary{arrows.meta,positioning}
\makeatletter
\@ifundefined{proof}{\newenvironment{proof}[1][Proof]{\par\noindent\textit{#1.}\ }{\hfill$\square$\par}}{}
\makeatother
\graphicspath{{figures/}}
\begin{document}
\title{Collective Quantum State Preparation Is Nonadditive}

\author{Lin Zhu}
\affiliation{Thrust of Artificial Intelligence, Information Hub,\\
The Hong Kong University of Science and Technology (Guangzhou),
Guangdong 511453, China}
\affiliation{Quantum Science Center of Guangdong--Hong Kong--Macao
Greater Bay Area, Shenzhen 518045, China}

\author{Ranyiliu Chen}
\email{chenranyiliu@quantumsc.cn}
\affiliation{Quantum Science Center of Guangdong--Hong Kong--Macao
Greater Bay Area, Shenzhen 518045, China}

\author{Xin Wang}
\affiliation{Thrust of Artificial Intelligence, Information Hub,\\
The Hong Kong University of Science and Technology (Guangzhou),
Guangdong 511453, China}

\author{Shenggen Zheng}
\email{zhengshenggen@quantumsc.cn}
\affiliation{Quantum Science Center of Guangdong--Hong Kong--Macao
Greater Bay Area, Shenzhen 518045, China}

\begin{abstract}
Quantum state preparation is a basic primitive of quantum computation and control, but its resource cost is usually analyzed for one target state or under an independent-copy construction. Here we study the many-copy problem under weighted Pauli control, where a Pauli-product rotation $e^{i\phi P}$ is charged by its rotation angle $|\phi|$, independently of the weight of $P$. We derive the exact single-qubit preparation cost and show that it is strictly nonadditive: a four-pulse collective protocol prepares two $T$ states at lower cost than two optimal one-copy protocols. We further identify finite regions of the Bloch sphere with collective saving, including trajectories that remain product throughout, showing that transient entanglement is not the only mechanism behind the advantage. Complementarily, a dimension-independent R\'enyi-entropy speed limit gives an extensive lower bound for every coherent target and an exponential fidelity penalty below a finite preparation rate. These results establish a many-copy geometry of state synthesis in which uncorrelated outputs can nevertheless benefit from collective control.
\end{abstract}
\maketitle

\textit{\textbf{Introduction.}}---Quantum state preparation provides the interface between a desired mathematical state and a controllable quantum device, and it is a basic subroutine in quantum algorithms, simulation, sensing, and fault-tolerant architectures. For a generic target, the preparation problem has traditionally been formulated in terms of discrete circuit resources, including one- and two-qubit gate counts, circuit depth, ancillary qubits, and space--time tradeoffs~\cite{Mottonen2005,ShendeBullockMarkov2006,PleschBrukner2011,Iten2016,Gui2024}. These works have developed increasingly efficient decompositions for arbitrary states and isometries. Complementary schemes exploit additional structure in the target, for example smooth functional amplitudes, while recent low-depth constructions show how workspace and parallelization can reduce the depth required to generate many target registers~\cite{ZylbermanDebbasch2024,Gui2024}. Such results quantify the complexity of a chosen gate architecture, but they do not directly address whether the \emph{integrated control cost} of producing several identical states is additive.

A complementary line of work views synthesis geometrically. In Nielsen's complexity geometry, and in related formulations of quantum optimal control, a circuit is replaced by a continuous path generated by admissible Hamiltonians and its cost is the length induced by a control norm~\cite{Nielsen2006,NielsenScience2006,NielsenControl2006,DowlingNielsen2008,BrownSusskind2019,Koch2022Control}. This perspective separates the geometry of reachable motion from a particular gate decomposition and makes it possible to ask whether nonlocal generators provide shortcuts. However, most geometric analyses concern a single unitary or a single state. The tensor-power problem introduces a distinct question: even when the target is the uncorrelated product $\psi^{\otimes n}$, can a collective Hamiltonian use the enlarged register more efficiently than $n$ independent controls?

Many-copy optimization and regularization appear throughout quantum information. Resource theories commonly distinguish one-shot costs from asymptotic rates, and collective processing can change the relevant conversion law~\cite{WinterYang2016,StreltsovAdessoPlenio2017,ChitambarGour2019}. Likewise, the $T$ state is central to fault-tolerant quantum computation, where collective coding and distillation strongly affect magic-state overhead~\cite{BravyiKitaev2005,BravyiHaah2012,Litinski2019}. The problem considered here is different from both settings. No class of free operations is imposed, Clifford operations are not assigned zero cost, and the resource is the total angle of Pauli-product rotations. Consequently, coherence monotones, stabilizer rank, and magic-state distillation rates do not determine the preparation cost studied below.

The weighted problem was posed by Knill, Ortiz, and Somma and recorded in the open-problem collection of Kr\"uger and Werner~\cite{KruegerWerner2005}; it has recently been restated as the collective-cost problem for tensor-power state preparation~\cite{QIQCOP2026}. A Pauli-product pulse $e^{i\phi P}$ is assigned charge $|\phi|$ independently of the weight of $P$, so the model isolates the geometric value of collective interactions. This weight independence is intentionally idealized: an angle saving in this metric need not equal a saving in laboratory time or energy unless the corresponding interaction strengths and implementation overheads are also comparable.

In this work, we establish that collective preparation is genuinely nonadditive in this model. First, we solve the unrestricted one-qubit problem exactly, obtaining a closed formula on the full Bloch sphere and a matching calibration lower bound. We then construct a four-pulse two-qubit protocol for $\ket T=(\ket0+e^{i\pi/4}\ket1)/\sqrt2$ whose cost per copy is strictly below the exact one-copy optimum. To our knowledge, this gives the first explicit rigorous separation $C_\infty(\psi)<C_1(\psi)$ for the unrestricted weighted Pauli-product rotation model. Beyond this example, we identify two mechanisms for saving: a finite collective conversion overhead that disappears after regularization, and interacting Hamiltonians that reproduce product-state tangents with a smaller Pauli coefficient norm. The latter mechanism yields certified saving regions even along trajectories that remain product throughout. Finally, we derive dimension-independent R\'enyi-entropy speed limits and a fidelity strong converse. They show that collective control can lower the linear coefficient of preparation cost, but cannot remove the extensive cost of creating computational-basis coherence.

\textit{\textbf{Preparation cost.}}---We now formulate the control model. The computational basis consists of binary-labeled states $\ket{x}$; $\ket{0^m}$ denotes the all-zero blank. Write $I$ for the identity and $X,Y,Z$ for the Pauli matrices, with $X\ket b=\ket{1-b}$, $Z\ket b=(-1)^b\ket b$, and $Y=iXZ$ for $b\in\{0,1\}$. A Pauli product is a tensor product of these operators, and its weight is the number of nonidentity factors. A pure ray identifies normalized state vectors that differ only by a global phase. For an $m$-qubit pure ray $\psi$ and $n$ identical outputs, define
\begin{equation}
 C_n(\psi)=\inf\left\{\sum_j|\phi_j|:
 \prod_j e^{i\phi_jP_j}\ket{0^{mn}}\doteq\ket\psi^{\otimes n}\right\},
 \label{eq:Cn}
\end{equation}
where $P_j\in\{I,X,Y,Z\}^{\otimes mn}$, $\doteq$ denotes equality up to global phase, and products act from right to left. The infimum is over all finite exact pulse sequences on the output register, with no additional workspace. Independent concatenation gives $C_{n+k}\le C_n+C_k$, and therefore the asymptotic cost per copy is the regularized quantity
\begin{equation}
 C_\infty(\psi)=\lim_{n\to\infty}\frac{C_n(\psi)}n
 =\inf_{n\ge1}\frac{C_n(\psi)}n\le C_1(\psi).
 \label{eq:Cinf}
\end{equation}
Thus one finite block with $C_k<kC_1$ already proves a persistent collective advantage for arbitrarily large tensor powers. In continuous control, $\dot\varphi=iH(t)\varphi$ with $H=\sum_Pu_PP$ and real coefficients $u_P$. The corresponding length is defined by the Pauli coefficient norm $\|H\|_{\mathcal P,1}=\sum_{P\ne I}|u_P|$ through $\int\|H\|_{\mathcal P,1}dt$. This is an $\ell_1$ control geometry in the Pauli basis, distinct from a local gate count because a high-weight Pauli direction is not additionally penalized. The equivalence between exact finite words and continuous controls is established in Appendix~\ref{sec:metric}.

\textit{\textbf{The optimal single-qubit benchmark.}}---Although the pure-state space of one qubit is only the Bloch sphere, the $\ell_1$ Pauli norm induces a nontrivial and anisotropic control geometry~\cite{BrownSusskind2019}. Represent the density operator as $(I+xX+yY+zZ)/2$, with unit Bloch vector $r=(x,y,z)$. An exact one-copy formula is essential here: without it, a two-copy protocol could only be compared with a particular independent circuit rather than with the true independent optimum.

\begin{theorem}[Exact one-copy cost]\label{thm:onecopy}
For a unit Bloch vector $r=(x,y,z)$, set $Y=\min(|x|,|y|)$ and $R=\sqrt{1-Y^2}$. With principal inverse-trigonometric branches,
\begin{equation}
 C_1(r)=\frac{\arcsin(\sqrt2Y)}{\sqrt2}
 +\frac12\left[\arccos\frac zR-\arcsin\frac YR\right].
 \label{eq:C1exact}
\end{equation}
\end{theorem}

The structure of the optimum follows from the symmetries of the Pauli control set. A Clifford unitary permutes signed Pauli products under conjugation, while the phase Clifford $S=\operatorname{diag}(1,i)$ fixes the blank. Powers of $S$, together with computational-basis complex conjugation, reduce the horizontal coordinates to the folded chamber $x\ge y\ge0$. A minimizing continuous path first moves along the diagonal section to $(y,y,\sqrt{1-2y^2})$ using equal-magnitude $X$ and $Y$ controls and then rotates about $Y$ to the target. These two pieces give the two terms in Eq.~\eqref{eq:C1exact}. The matching lower bound is a calibration whose variation under every Pauli pulse is at most the charged angle. The detailed proof and a family of exact finite circuits approaching the optimum can be found in Appendix~\ref{sec:C1}. Hence $C_1$ is the unrestricted exact-circuit infimum, not the cost of a restricted ansatz.

\textit{\textbf{Two copies cost less.}}---We next consider the $T$ state, $\ket T=(\ket0+e^{i\pi/4}\ket1)/\sqrt2$. This state is a canonical nonstabilizer resource in fault-tolerant quantum computation~\cite{BravyiKitaev2005,BravyiHaah2012,Litinski2019}, but here it is used only as a target in the angle-cost geometry: Clifford transformations themselves are not free. We abbreviate tensor products as $XX=X\otimes X$, and similarly for other two-qubit Paulis. Theorem~\ref{thm:onecopy} gives the exact independent benchmark $C_1(T)=\pi/(2\sqrt2)$.
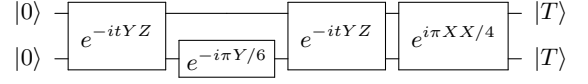
\begin{figure}[t]
\[
\Qcircuit @C=5pt @R=6pt {
 \lstick{\ket0} & \multigate{1}{e^{-itYZ}} & \qw & \multigate{1}{e^{-itYZ}} & \multigate{1}{e^{i\pi XX/4}} & \rstick{\ket T}\qw\\
 \lstick{\ket0} & \ghost{e^{-itYZ}} & \gate{e^{-i\pi Y/6}} & \ghost{e^{-itYZ}} & \ghost{e^{i\pi XX/4}} & \rstick{\ket T}\qw
}
\]
\caption{The four-pulse circuit prepares two $T$ states exactly.}
\label{fig:circuit}
\end{figure}

\begin{theorem}[Collective $T$-state shortcut]\label{thm:shortcut}
For $t=\tfrac12\arctan(1/\sqrt2)$, the circuit
\begin{equation}
 U_2=e^{i\pi XX/4}e^{-itYZ}e^{-i\pi IY/6}e^{-itYZ}
 \label{eq:U2}
\end{equation}
maps $\ket{00}$ to $\ket T^{\otimes2}$ and obeys
\begin{align}
 C_\infty(T)&\le\frac{C_2(T)}2\le c_T,\nonumber\\
 c_T&=\frac12\left(\arctan\frac1{\sqrt2}+\frac{5\pi}{12}\right),\nonumber\\
 C_1(T)-C_\infty(T)&\ge\frac\pi{2\sqrt2}-c_T>0.
 \label{eq:Tgap}
\end{align}
\end{theorem}

The four-pulse word is therefore not merely shorter than a convenient product circuit; it beats twice the exact one-copy optimum. The protocol in Fig.~\ref{fig:circuit} passes through correlated states. For amplitudes $(a,b,c,d)$, the concurrence $2|ad-bc|$ vanishes exactly on product states~\cite{Wootters1998}. Immediately before the final pulse, the state is $\ket{00}/\sqrt2+(\ket{01}+\ket{10})/2$, whose concurrence is $1/2$, and the final $XX$ rotation converts it into the product target. The strict comparison follows from $c_T<\pi/3<\pi/(2\sqrt2)$. The detailed gate calculation and the continuity argument showing that the separation persists in a neighborhood of $T$ can be found in Appendix~\ref{sec:T}.

\begin{figure}[h]
\centering
\vspace{2pt}
\includegraphics[width=\columnwidth]{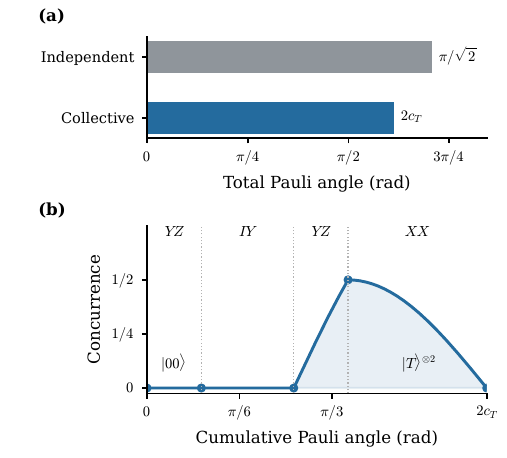}
\caption{\textbf{Collective preparation of a product target at reduced cost.}  (a) Its total angle $2c_T$ is below the optimal independent cost $2C_1(T)=\pi/\sqrt2$. (b) Concurrence along the exact circuit. Entanglement is generated transiently and vanishes at the product endpoint.}
\label{fig:shortcut}
\end{figure}

The Pauli orbit of a ray consists of its images under $I,X,Y,Z$. For $T$, the Clifford $S$ cycles all four orbit states and preserves their one-copy costs. The same collective separation therefore holds relative to the best independent preparation throughout this orbit.

\textit{\textbf{Two mechanisms for collective saving.}}---The $T$ circuit shows that correlations can mediate a shortcut, but it does not imply that correlations are necessary. The exact one-qubit formula allows us to separate two mechanisms. We call a Bloch vector northern when $z\ge0$ and southern when $z<0$.

\begin{proposition}[Collective orbit conversion]\label{prop:orbit}
For any one-qubit Pauli $Q$, $C_\infty(Q\psi)=C_\infty(\psi)$. Every southern Bloch vector, $z<0$, consequently satisfies
\begin{equation}
 C_1(r)-C_\infty(r)\ge\arcsin\frac{|z|}{\sqrt{1-Y^2}}>0.
 \label{eq:south}
\end{equation}
\end{proposition}

The reason is an amortization effect. The transformation $Q^{\otimes n}$ is itself one allowed Pauli product, so converting all $n$ copies costs at most $\pi/2$, independent of $n$. Combining this operation with complex conjugation relates $r$ to its northern reflection $(x,y,-z)$. Their regularized costs coincide, while their one-copy costs differ by Eq.~\eqref{eq:south}. A bounded collective conversion cost can therefore disappear entirely at the per-copy level; the detailed derivation can be found in Appendix~\ref{sec:T}.

Northern targets reveal a different mechanism: collective interactions can reproduce a desired \emph{product-state tangent} with a smaller coefficient norm than the sum of local generators. Let $x>y>0$, $z>0$, write $R_r=xX+yY+zZ$ and $S_2=XX+YY+ZZ$, and consider
\begin{equation}
 H_X=X\otimes R_r+R_r\otimes X-\frac{R_r\otimes R_r}{x}
 +\frac{\min(y^2,z^2)}xS_2.
 \label{eq:HX}
\end{equation}
At $\psi^{\otimes2}$ this Hamiltonian generates the same instantaneous ray motion as $X_1+X_2$, because $R_r\psi=\psi$ and $S_2\psi^{\otimes2}=\psi^{\otimes2}$. Its per-copy charge is $Q_X=[1-\min(y^2,z^2)+2yz]/(2x)$, whereas the incoming local $+iX$ tangent increases $C_1$ at rate $p=(xy+z\sqrt{1-2y^2})/(1-y^2)$. Whenever $Q_X<p$, a short terminal segment of the independent optimum can be replaced by a cheaper collective flow.

A particularly simple integrated construction applies near the diagonal of the northern Bloch sphere.

\begin{proposition}[Northern saving region]\label{prop:northern}
Every northern qubit ray with $Y>2/3$ satisfies
\begin{equation}
 C_2(r)<2C_1(r).
 \label{eq:northern-saving}
\end{equation}
\end{proposition}

To see the mechanism, consider the diagonal family
$\ket{\psi_\theta}=\cos\theta\ket0+e^{i\pi/4}\sin\theta\ket1$.
Appendix~\ref{sec:diagonal-flow} constructs a two-copy Hamiltonian that
follows exactly the product tangent of the one-copy optimal path.  Its
cost per copy is $\sin2\theta+\sqrt2\cos2\theta$, and comparison with
the one-copy rate gives a saving function $g(\theta)$ satisfying
$g(\theta_0)=g'(\theta_0)=0$ and $g''(\theta)>0$ for
$\theta>\theta_0=\arctan(1/\sqrt2)$.  Since $Y=2/3$ at
$\theta=\theta_0$, every diagonal point with $Y>2/3$ has a strict
two-copy deficit.  Appending the remaining one-copy minimizing $Y$-rotation independently to both outputs preserves the same deficit along the entire fixed-$Y$ branch. The propagation argument is given in Appendix~\ref{sec:regions}, completing the proof of Proposition~\ref{prop:northern}.

The equator exhibits a second, quantitatively distinct saving region.
Here no single tangent certificate covers the whole interval; instead,
three overlapping certificates are combined.

\begin{proposition}[Equatorial saving arc]\label{prop:equatorial}
For an equatorial ray
\begin{equation}
 r_\phi=(\cos\phi,\sin\phi,0),
 \qquad
 \arcsin\frac{19}{80}\le\phi\le\frac{\pi}{4},
\end{equation}
one has
\begin{equation}
 C_2(r_\phi)<2C_1(r_\phi).
 \label{eq:equatorial-saving}
\end{equation}
\end{proposition}

On the reduced equator, $Y=\sin\phi$.  Appendix~\ref{sec:equatorial-arc}
proves an $X$-tangent deficit for $19/80\le Y\le1/4$ at a fixed
interior point of the one-copy minimizing branch and propagates that
deficit to the equator.  A second $X$-tangent estimate covers
$1/4\le Y\le1/\sqrt3$.  For $Y>1/\sqrt3$, the $Y$-tangent
representative has per-copy charge $Q_Y=1/(2Y)<1$, which gives a
strict deficit up to the diagonal endpoint $T$.  These intervals join
to give the complete arc in Proposition~\ref{prop:equatorial}; the
exact inequalities are proved in Appendix~\ref{sec:equatorial-arc}.

Both propositions therefore admit collective trajectories that remain
on the product-state manifold throughout.  The interaction terms cancel
motion normal to that manifold while retaining the desired tangent at
smaller $\ell_1$ charge.  Transient entanglement can mediate a shortcut,
as in the $T$-state circuit, but it is not necessary for nonadditivity.

\begin{figure}[t]
\centering
\includegraphics[width=\columnwidth]{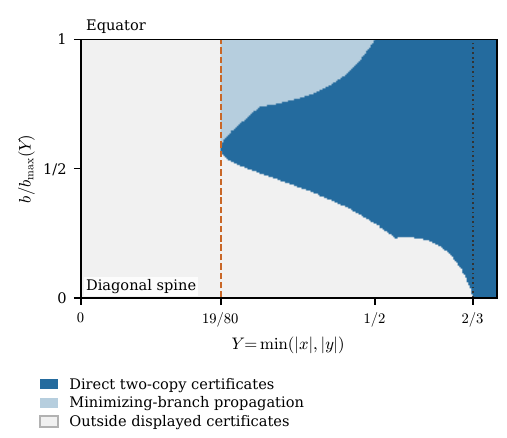}
\caption{\textbf{Certified collective-saving regions on the northern Bloch sphere.} The folded hemisphere is parameterized by $y=Y$ and the final $Y$-rotation angle $b$, normalized by its equatorial value $b_{\max}(y)$. Dark shading denotes direct tangent certificates and light shading their propagation along independently minimizing suffixes. Proposition~\ref{prop:northern} gives the region $Y>2/3$, while Proposition~\ref{prop:equatorial} gives the certified equatorial arc beginning at $Y=19/80$. Gray points are not classified by the present certificates.}
\label{fig:regions}
\end{figure}

For Fig.~\ref{fig:regions}, the one-copy minimizing branch is a $-iY$ rotation from $(y,y,\sqrt{1-2y^2})$, with $C_1=A(y)+b$, $A(y)=\arcsin(\sqrt2y)/\sqrt2$, and $b_{\max}=\pi/4-\tfrac12\arcsin[y/\sqrt{1-y^2}]$. Once a point has a strict finite-block saving, appending the one-copy optimal suffix to each output preserves the same deficit at all later points of the branch. This propagation turns a local tangent advantage into a finite target region.

\textit{\textbf{Coherence imposes an extensive floor.}}---Nonadditivity does not mean that collective preparation can make a coherent target asymptotically free. In the resource theory of coherence, deleting off-diagonal entries in the computational basis assigns a pure state with probabilities $p_x=|\langle x|\psi\rangle|^2$ the relative entropy of coherence $H(p)=-\sum_xp_x\log p_x$~\cite{Baumgratz2014,WinterYang2016,StreltsovAdessoPlenio2017}. Here the operational question is different: we ask how quickly the allowed Pauli control can create this population entropy starting from a basis state. For $\alpha\ne1$, write the R\'enyi entropy as $H_\alpha(p)=(1-\alpha)^{-1}\log\sum_xp_x^\alpha$, with natural logarithms.

\begin{theorem}[Entropy and fidelity bounds]\label{thm:entropy}
For $1/2<\alpha<1$, define
\begin{equation}
 \kappa_\alpha=\frac{2\alpha}{1-\alpha}
 \sup_{u\ge0}\frac{\sinh[(1-\alpha)u]}{\cosh(\alpha u)}.
 \label{eq:kappa}
\end{equation}
For evolving probabilities $q_x(t)=|\langle x|\varphi(t)\rangle|^2$ under a Pauli flow with $|u_P|=1$, $|\dot H_\alpha(q)|\le\kappa_\alpha$, independently of register size. Hence
\begin{equation}
 C_n(\psi)\ge n\frac{H_\alpha(p)}{\kappa_\alpha}.
 \label{eq:entropy}
\end{equation}
For a length-$L$ circuit with output $\varphi_n$, its squared fidelity $F_n=|\langle\psi^{\otimes n}|\varphi_n\rangle|^2$ obeys
\begin{equation}
 F_n\le\exp\left[-\frac{1-\alpha}{\alpha}
 \left(nH_\beta(p)-\kappa_\alpha L\right)\right],
 \quad\beta=\frac{\alpha}{2\alpha-1}.
 \label{eq:converse}
\end{equation}
\end{theorem}

The dimension-independent bound follows from a simple structural property of Pauli dynamics. Any nondiagonal Pauli operator partitions the computational basis into disjoint pairs, and the probability current within a pair of populations $a$ and $b$ is bounded by $2\sqrt{ab}$. Optimizing the resulting R\'enyi-entropy production over a single pair and summing over all pairs gives a universal speed limit that is independent of the register size. H\"older's inequality then turns this entropy constraint into the fidelity bound of Eq.~\eqref{eq:converse}. In particular, allowing a Pauli interaction to act collectively on all copies does not increase the maximal entropy-production rate per unit control cost. The full derivation is given in Appendix~\ref{sec:entropy}. Unlike conventional quantum speed limits, which constrain evolution through energy, variance, geometric distance, or related dynamical quantities~\cite{MargolusLevitin1998,DeffnerCampbell2017,Marvian2016,MondalDattaSazim2016,FunoShiraishiSaito2019,Bukov2019}, the present bound limits the rate at which computational-basis coherence can be generated per unit Pauli-control length.

Taking the R\'enyi parameter to the Shannon-entropy limit gives a finite threshold for the preparation cost per copy. Any family of protocols whose rate remains strictly below this threshold has exponentially vanishing fidelity in the many-copy limit. Consequently, every coherent finite-dimensional target requires a cost that grows linearly with the number of copies, even when a fixed or vanishing preparation error is allowed. Computational-basis states behave qualitatively differently: all copies of any nonzero basis state can be produced by a single collective Pauli rotation of total cost $\pi/2$, independent of the number of copies. Hence the regularized cost vanishes exactly for computational-basis rays, whereas every coherent target retains a strictly positive asymptotic cost per copy. The explicit threshold and the corresponding approximate-preparation bounds are derived in Appendix~\ref{sec:growth}.

\textit{\textbf{Discussion.}}---Our results show that the tensor-product structure of the target does not determine the tensor-product structure of an optimal preparation. This distinguishes the present effect from low-depth parallel state-preparation schemes, where the gain is measured in circuit depth or space--time allocation~\cite{Gui2024}. Here the resource itself is nonadditive: the total Pauli rotation angle needed per copy can decrease when the copies are synthesized jointly. The exact one-qubit benchmark is essential for making this statement sharp, while the $T$ shortcut, orbit conversion, and product-tangent constructions show that the advantage has more than one geometric origin.

The entropy bounds give the complementary many-copy limitation. In contrast with the one-copy Fubini--Study distance, they remain extensive under tensor powers and are insensitive to the weight of an individual Pauli generator. This identifies computational-basis coherence as an obstruction to zero asymptotic rate, while leaving room for collective reductions of the coefficient. The resulting picture parallels a common theme in quantum resource theories---regularization can reveal a different asymptotic cost from the one-shot value~\cite{ChitambarGour2019}---but the responsible resource here is the control geometry generated by weighted Pauli products rather than a conventional free-operation structure.

Several parts of the original collective-cost problem remain open. The complete saving boundary on the northern Bloch sphere and the exact value of $C_\infty$ are not determined. In particular, for $\ket+=(\ket0+\ket1)/\sqrt2$, we prove $C_2(\ket+)=\pi/2$, while the unrestricted regularized value is only bounded by $(\log2)/\kappa_1\le C_\infty(\ket+)\le\pi/4$. The derivation of these bounds is given in Appendices~\ref{sec:endpoint} and \ref{sec:deficiency}. We also obtain several endpoint obstructions. The additive floor $N\pi/4$ holds for pulses on the $(\pi/8)\mathbb Z$ angle grid, for trajectories whose pulse-boundary stabilizer deficiency never exceeds two, and for paths whose population-norm vector remains in the positive product cube $\bigotimes_j(\cos\theta_j,\sin\theta_j)$ with $0\le\theta_j\le\pi/4$. The corresponding proofs are given in Appendices~\ref{sec:grid}, \ref{sec:matching}, and \ref{sec:deficiency}, respectively. Any unrestricted shortcut to the $\ket+^{\otimes N}$ endpoint must therefore leave all of these protected regimes. Branch interference beyond them remains the central obstruction to a global additivity theorem, and its scope is discussed further in Appendix~\ref{sec:branches}.

More generally, exact all-block additivity requires a calibration that remains valid on correlated states of arbitrary register size. A dual characterization of $C_\infty(\psi)$ as the maximum of $F(\psi)$ over real tensor-additive, Pauli-1-Lipschitz calibrations can be found in Appendix~\ref{sec:dual}. Evaluating this dual for specific coherent states, determining the full saving boundary, extending the model to weight-dependent interaction penalties, and closing the fixed-error asymptotic rate are natural next steps. These questions connect state synthesis, quantum control, and many-copy resource accounting in a setting where the target remains simple but the optimal path need not.

\textit{\textbf{AI usage statement.}}---
Generative artificial-intelligence tools were used extensively in this work. A substantial majority of the initial analytical derivations, proof constructions, intermediate calculations, and manuscript text was generated under human direction. The authors formulated the research questions, provided problem-specific constraints and mathematical context, selected which generated approaches to develop, critically examined the resulting arguments, and independently checked the final claims, calculations, and proofs before inclusion in the manuscript. All AI-generated material was reviewed and revised by the authors, who take full responsibility for the correctness and content of the work.

\textbf{\textit{Acknowledgments.}}---This work was partially supported by the National Key R\&D Program of China (Grant No.~2024YFB4504004), the National Natural Science Foundation of China (Grant No.92576114), the Guangdong Provincial Quantum Science Strategic Initiative (Grant No.~GDZX2403008, GDZX2503001), and the Guangdong Provincial Key Lab of Integrated Communication, Sensing and Computation for Ubiquitous Internet of Things (Grant No.~2023B1212010007).
\bibliography{ref}

\newpage
\onecolumngrid
\vspace{2cm}
\setcounter{secnumdepth}{2}
\setcounter{section}{0}
\setcounter{subsection}{0}
\setcounter{equation}{0}
\setcounter{figure}{0}
\setcounter{table}{0}
\setcounter{theorem}{0}
\renewcommand{\thesection}{\Roman{section}}
\renewcommand{\theequation}{S\arabic{equation}}
\renewcommand{\thefigure}{S\arabic{figure}}
\renewcommand{\thetable}{S\arabic{table}}
\renewcommand{\thetheorem}{S\arabic{theorem}}
\renewcommand{\theHsection}{supp.\arabic{section}}
\renewcommand{\theHequation}{supp.\arabic{equation}}
\renewcommand{\theHfigure}{supp.\arabic{figure}}
\renewcommand{\theHtable}{supp.\arabic{table}}
\providecommand{\theHtheorem}{\arabic{theorem}}
\renewcommand{\theHtheorem}{supp.\arabic{theorem}}
\begin{center}
{\large\bfseries Supplemental Material for
``Collective Quantum State Preparation Is Nonadditive''}
\end{center}
In this Supplemental Material, we provide the detailed derivations and endpoint restrictions supporting the results of the Letter. Throughout, states are pure rays, Pauli generators are Hermitian, logarithms are natural, and a pulse is $e^{+itP}$ with charge $|t|$. A Bloch vector therefore evolves as $\dot r=-2e_P\times r$. Whenever the opposite field convention is used, it is stated explicitly. The supplemental results do not assume an unrestricted no-saving classification.

Section~\ref{sec:metric} establishes the path metric, the equivalence between continuous controls and exact finite words, and the regularized limit. Section~\ref{sec:C1} derives the exact one-copy benchmark on the full Bloch sphere. Section~\ref{sec:T} proves the collective $T$-state shortcut and orbit conversion, while Sec.~\ref{sec:regions} develops the product-tangent constructions and the certified northern saving regions. Sections~\ref{sec:entropy} and \ref{sec:growth} establish the entropy speed limit, fidelity strong converse, and exact-versus-approximate growth statements. The remaining sections develop the tensor-additive dual characterization and the endpoint obstructions used in the Discussion of the Letter.

\section{Path metric, exact finite words, and regularization}\label{sec:metric}
For fixed register size, let $d(a,b)$ denote the infimum of the total angle of finite Pauli words taking ray $a$ to ray $b$. Inversion and concatenation give symmetry and the triangle inequality. The Fubini--Study distance is $d_{\rm FS}(a,b)=\arccos|\langle a|b\rangle|$ for normalized representatives. The Fubini--Study speed of a unit Pauli flow is $\sqrt{\langle P^2\rangle-\langle P\rangle^2}\le1$, so
\begin{equation}
 d_{\rm FS}(a,b)\le d(a,b).
\end{equation}
Tensoring independent transitions gives $d(a\otimes c,b\otimes e)\le d(a,b)+d(c,e)$. Conjugating a circuit by a Clifford that fixes the blank preserves preparation cost; it does not assign zero charge to inserting that Clifford. Complex conjugation also preserves preparation cost, because $e^{itP}$ becomes $e^{-it\overline P}$ and $\overline P$ is a signed Pauli.

\begin{lemma}[Continuous and finite-word infima]\label{lem:control}
On a fixed register, integrable control $\dot U=i\sum_Pu_P(t)PU$ with length $\int\sum_P|u_P(t)|dt$ gives the same transition infimum as finite Pauli words.
\end{lemma}
\begin{proof}
Finite words are particular controls. Conversely, approximate an integrable control in $L^1$ by piecewise-constant controls and split each constant interval by a product formula. The resulting words converge in operator norm and their angle sums converge to the integrated coefficient norm. The ordered product map of all $4^m-1$ nonidentity Pauli one-parameter subgroups has nonsingular derivative at the identity, since these generators form a basis of $\mathfrak{su}(2^m)$. It is a local coordinate chart. The residual unitary tending to the identity can therefore be corrected exactly with angle sum tending to zero. Global phases are irrelevant. This proves equality of the exact infima, not merely an approximate-endpoint statement.
\end{proof}

Every finite pure ray has finite preparation cost, since the Pauli axes span the full unitary Lie algebra. Concatenation implies subadditivity of $C_n$. Fekete's lemma gives
\begin{equation}
 C_\infty(\psi)=\inf_n\frac{C_n(\psi)}n,\qquad
 C_n(\psi)=nC_\infty(\psi)+o(n).
 \label{eq:regularization}
\end{equation}
In particular, $C_1-C_\infty=\sup_n[C_1-C_n/n]$; a positive gap is witnessed at some finite block.

\section{Exact one-qubit cost on the full Bloch sphere}\label{sec:C1}
\begin{theorem}[One-copy formula]\label{thm:C1}
For a unit Bloch vector $r=(x,y,z)$, put $Y=\min(|x|,|y|)$ and $R=\sqrt{1-Y^2}$. Then
\begin{equation}
 C_1(r)=A(Y)+\frac12\left[\arccos(z/R)-\arcsin(Y/R)\right],
 \qquad A(Y)=\frac{\arcsin(\sqrt2Y)}{\sqrt2}.
 \label{eq:W}
\end{equation}
\end{theorem}
\begin{proof}
Conjugation by $S=\operatorname{diag}(1,i)$ and complex conjugation reduce to $x\ge y\ge0$ without changing $z$ or cost. Write $q=\sqrt{1-2y^2}$, $R^2=1-y^2$, and
\begin{equation}
 W(r)=\frac\pi4+B(y)-\frac12\arcsin(z/R),\qquad
 B(y)=A(y)-\tfrac12\arcsin(y/R).
\end{equation}
In the open chamber $x>y>0$, $B'(y)=q/(2R^2)$ and
\begin{equation}
 W_y=\frac{xq-yz}{2R^2x},\qquad W_z=-\frac1{2x}.
\end{equation}
For the Bloch fields $2e_j\times r$ (opposite to the pulse convention), the derivatives are
\begin{equation}
 \partial_XW=-p,\quad\partial_YW=1,\quad\partial_ZW=D,\qquad
 p=\frac{xy+zq}{R^2},\quad D=\frac{xq-yz}{R^2}.
\end{equation}
Their identity $p^2+D^2=1$ follows from $x^2+z^2=y^2+q^2=R^2$. Thus every signed Pauli field has absolute slope at most one.

The folded expression is globally continuous and locally Lipschitz. Near the north pole it equals $A(y)+[\arcsin(x/R)-\arcsin(y/R)]/2$; near the south pole it equals $\pi/2+A(y)-[\arcsin(x/R)+\arcsin(y/R)]/2$. These are locally Lipschitz in the folded horizontal coordinates. At the diagonal equator the cancellation in $B'(y)$ removes the apparent singularity. A generic Pauli orbit meets seams only at isolated times, so its slope integrates to $|W(e^{itP}r)-W(r)|\le|t|$. Approximate exceptional starting rays by generic ones and use continuity to obtain the same inequality everywhere. Since $W(0,0,1)=0$, every preparation word costs at least $W$.

For the upper bound, follow
\begin{equation}
 r(\alpha)=\left(\frac{\sin\alpha}{\sqrt2},
 \frac{\sin\alpha}{\sqrt2},\cos\alpha\right),\qquad
 0\le\alpha\le\arcsin(\sqrt2y).
\end{equation}
Equal-magnitude, opposite-sign $X,Y$ coefficients generate this curve with cost $A(y)$ and endpoint $(y,y,q)$. A $-iY$ rotation then increases the $XZ$-plane Bloch angle from $a_0=\arcsin(y/R)$ to $a_1=\arccos(z/R)$, costing $(a_1-a_0)/2$. The inequality $x\ge y$ ensures $a_1\ge a_0$ even for southern targets. Lemma~\ref{lem:control} transfers this matching control length to exact finite-word infima.
\end{proof}

On a one-Pauli meridian, $Y=0$ and $C_1=\tfrac12\arccos z$. On the reduced equator, $r=(\cos\phi,\sin\phi,0)$ with $0\le\phi\le\pi/4$, the formula becomes
\begin{equation}
 C_1(\phi)=\frac{\arcsin(\sqrt2\sin\phi)}{\sqrt2}
 +\frac\pi4-\frac12\arcsin(\tan\phi).
\end{equation}
In particular, $C_1(T)=\pi/(2\sqrt2)$.

\subsection{Parameterized exact finite circuits approaching the optimum}\label{sec:finiteC1}
In the folded chamber $x\ge y\ge0$, let $\alpha=\arcsin(\sqrt2y)$ and $b=[\arccos(z/R)-\arcsin(y/R)]/2$. For any positive integer $N$, put
\begin{equation}
 s_j=\frac{\sin(j\alpha/N)}{\sqrt2},\qquad R_j=\sqrt{1-s_j^2},\qquad 0\le j\le N,
\end{equation}
and define
\begin{align}
 a_j&=\frac12\left[\arcsin\frac{s_j}{R_{j-1}}-\arcsin\frac{s_{j-1}}{R_{j-1}}\right],\\
 b_j&=\frac12\left[\arcsin\frac{s_j}{R_j}-\arcsin\frac{s_{j-1}}{R_j}\right].
\end{align}
The circuit
\begin{equation}
 U_N=e^{-ibY}(e^{-ib_NY}e^{ia_NX})\cdots(e^{-ib_1Y}e^{ia_1X})
 \label{eq:finiteC1}
\end{equation}
prepares the target ray exactly for every $N$. Its cost $L_N=b+\sum_{j=1}^N(a_j+b_j)$ converges to $C_1(r)$.

Each pulse pair first increases $y$ from $s_{j-1}$ to $s_j$ at fixed $x=s_{j-1}$, then increases $x$ to $s_j$ at fixed $y=s_j$. Thus it maps one diagonal point exactly to the next, and the final $Y$ rotation reaches $r$. For $u=s_{j-1}$ and $w=s_j$, its pair cost is
\begin{equation}
 a_j+b_j=\frac12\int_u^w\left[
 \frac1{\sqrt{1-u^2-t^2}}+\frac1{\sqrt{1-w^2-t^2}}\right]dt.
\end{equation}
Away from $y=1/\sqrt2$, the sums converge to $\int_0^y(1-2t^2)^{-1/2}dt=A(y)$. At the endpoint, the first integral is at most the corresponding increment of $A$; the second obeys
\begin{equation}
 \int_u^w\frac{dt}{\sqrt{1-w^2-t^2}}
 \le\sqrt2\int_u^w\frac{dt}{\sqrt{1-2t^2}}.
\end{equation}
To prove this inequality, fix $w\le1/\sqrt2$. Both sides vanish at $u=w$, and their difference has nonpositive derivative with respect to $u$, since $1-w^2-u^2\ge(1-2u^2)/2$. These integral bounds make the endpoint tail vanish uniformly as the mesh shrinks. Therefore $L_N\to A(y)+b=C_1(r)$ even at the $T$ endpoint. Global calibration gives $L_N\ge C_1(r)$ for every $N$. The family approaches the exact infimum without asserting finite-word attainment or monotonicity in $N$.

\section{The exact two-copy shortcut and robustness}\label{sec:T}
Put $t=\tfrac12\arctan(1/\sqrt2)$ and $U_2=e^{i\pi XX/4}e^{-itYZ}e^{-i\pi IY/6}e^{-itYZ}$. Chronological evolution from the blank gives
\begin{align}
 v_1&=\cos t\ket{00}+\sin t\ket{10},\\
 v_2&=\tfrac{\sqrt3}2\cos t\ket{00}+\tfrac12\cos t\ket{01}
       +\tfrac{\sqrt3}2\sin t\ket{10}+\tfrac12\sin t\ket{11},\\
 v_3&=\tfrac{\sqrt3}2\cos(2t)\ket{00}+\tfrac12\ket{01}
       +\tfrac{\sqrt3}2\sin(2t)\ket{10}
       =\tfrac1{\sqrt2}\ket{00}+\tfrac12(\ket{01}+\ket{10}).
\end{align}
Since $\sin(2t)=1/\sqrt3$ and $\cos(2t)=\sqrt{2/3}$, the last $XX$ pulse gives
\begin{equation}
 U_2\ket{00}=\tfrac12\left(\ket{00}+e^{i\pi/4}\ket{01}
 +e^{i\pi/4}\ket{10}+i\ket{11}\right)=\ket T^{\otimes2}.
\end{equation}
The cost is $2t+5\pi/12=2c_T$. For amplitudes $(a,b,c,d)$, concurrence is $2|ad-bc|$; hence $v_3$ has concurrence $1/2$ and the final state has zero concurrence.

For any one-qubit Pauli $Q$, applying $Q^{\otimes n}$ is a single Pauli pulse of cost at most $\pi/2$, giving $|C_n(Q\psi)-C_n(\psi)|\le\pi/2$ and $C_\infty(Q\psi)=C_\infty(\psi)$. The four rays $Q\ket T$ are the four equatorial phases $(2k+1)\pi/4$ and are cycled by $S$. Conjugating by $S^{\otimes n}$ fixes the blank and permutes signed Pauli generators, so all four costs agree at every $n$.

For $r=(x,y,z)$, complex conjugation gives $(x,-y,z)$ and the Pauli $X$ maps this to $(x,y,-z)$. Both operations preserve $C_\infty$. For $z<0$, the one-copy formula therefore gives
\begin{equation}
 C_1(r)-C_1(x,y,-z)=\arcsin\frac{|z|}{\sqrt{1-Y^2}}.
\end{equation}
Since $C_\infty(x,y,-z)\le C_1(x,y,-z)$, this proves the southern bound in Proposition~\ref{prop:orbit} of the Letter.

If $d$ is the one-qubit transition metric, appending independent transitions implies
\begin{equation}
 |C_1(\psi)-C_1(\varphi)|\le d(\psi,\varphi),\qquad
 |C_\infty(\psi)-C_\infty(\varphi)|\le d(\psi,\varphi).
\end{equation}
A round shortest rotation of Fubini--Study angle $\delta$ has coefficient $\ell_1$ length at most $\sqrt3\delta$, so Lemma~\ref{lem:control} gives $d\le\sqrt3d_{\rm FS}$. Therefore the gap is positive whenever
\begin{equation}
 d_{\rm FS}(\psi,T)<\frac{\pi/(2\sqrt2)-c_T}{2\sqrt3}.
\end{equation}

\section{Two-copy tangent constructions and certified regions}\label{sec:regions}
For the northern folded sphere, use
\begin{align}
 r_y(b)&=(y\cos2b+q\sin2b,\ y,\ q\cos2b-y\sin2b),\qquad q=\sqrt{1-2y^2},\\
 0\le b&\le b_{\max}(y)=\frac\pi4-\frac12\arcsin\frac y{\sqrt{1-y^2}},
 \qquad C_1(r_y(b))=A(y)+b.
 \label{eq:branch}
\end{align}
Along a fixed-$y$ branch, a local $-iY$ rotation increases the one-copy cost by precisely its angle. Appending that suffix on each output propagates any finite-block saving to every later point of the branch.

\begin{lemma}[Strict tangent deficit]\label{lem:deficit}
Suppose a Hermitian $n$-qubit Pauli Hamiltonian $H$ generates at $\psi^{\otimes n}$ the product-ray tangent of a one-copy generator $G$. Suppose $C_1$ is differentiable along the incoming $+iG$ path, with derivative $d>0$, and $\|H\|_{\calP,1}<nd$, where the norm sums absolute nonidentity Pauli coefficients. Then $C_n(\psi)<nC_1(\psi)$.
\end{lemma}
\begin{proof}
For a Pauli $G$, conjugating $H$ by $(e^{isG})^{\otimes n}$ follows the incoming product path exactly, since it commutes with $\sum_jG_j$. Its coefficient norm is continuous. Independently prepare the ray at $s=-\eta$ and append this collective path. The cost is at most $nC_1(\psi)-n\eta d+\eta\|H\|_{\calP,1}+o(\eta)$, strictly below $nC_1(\psi)$ for small positive $\eta$. More generally, for a specified smooth product curve with a continuously chosen tangent representative, integrate that representative over a short terminal segment. Lemma~\ref{lem:control} and near-optimal exact input preparations preserve the strict margin. The statement concerns an infimum; it does not assume a minimum finite preparation word.
\end{proof}

\subsection{Axis tangents}
At $r=(x,y,z)$ with $x>y>0$, set $R_r=xX+yY+zZ$ and $S_2=XX+YY+ZZ$. Since $R_r\psi=\psi$ and $S_2=2\operatorname{SWAP}-I$, where $\operatorname{SWAP}$ exchanges the two tensor factors, fixes $\psi^{\otimes2}$, the Hamiltonian in the Letter satisfies
\begin{equation}
 H_X\psi^{\otimes2}=(X_1+X_2)\psi^{\otimes2}
 +\frac{m_X-1}{x}\psi^{\otimes2},\qquad m_X=\min(y^2,z^2).
\end{equation}
Its nonzero coefficients are
\begin{equation}
 XX:x+m_X/x,\quad YY:(m_X-y^2)/x,\quad ZZ:(m_X-z^2)/x,
 \quad YZ=ZY:-yz/x.
\end{equation}
Thus $\|H_X\|_{\calP,1}=2Q_X=(1-m_X+2yz)/x$. For the pulse-sign $+iX$ tangent, $\dot C_1=p=\cos2b$. Lemma~\ref{lem:deficit} proves saving when $Q_X<p$.

Likewise $H_Y=Y\otimes R_r+R_r\otimes Y-R_r\otimes R_r/y+(m_Y/y)S_2$, $m_Y=\min(x^2,z^2)$, has norm $2Q_Y=(1-m_Y+2xz)/y$ and generates the $Y$ product tangent up to phase. Use $-H_Y$ for the increasing-cost $-iY$ direction. The condition $Q_Y<1$ therefore proves saving. These are feasible Hamiltonian costs, not assertions of optimality over all representatives.

\subsection{A transverse product curve}
For $x>y>0,z>0,q>0$, let $R^2=1-y^2$ and define
\begin{equation}
 c=\frac{z-qxy}{2qzR^2},\quad u=-\frac{R^2}{2xz},\quad
 v=\frac{y^2-z^2}{2xz},\quad w=\frac y{2x}.
\end{equation}
The Hamiltonian $H_A=c(YI+IY)+uXX+vYY+w(YZ+ZY)$, with Schr\"odinger sign $-iH_A$, follows $r_y(b)^{\otimes2}$ as $y$ varies at fixed $b$. Indeed, its effective local field is $(ux,c+vy+wz,wy)$ and twice its cross product with $r$ equals
\begin{equation}
 \partial_y r_y(b)=\left(\frac{z-qxy}{qR^2},\ 1,
 -\frac{yz+x/q}{R^2}\right).
\end{equation}
Its symmetric coupling matrix $J$ obeys
\begin{equation}
 J+\frac z{2x}I=ra^T+ar^T,\qquad
 a=\left(-\frac1{4z},\frac y{4xz},\frac1{4x}\right).
\end{equation}
For a complex tangent vector $m$ with $m\cdot r=m\cdot m=0$, this gives $m^TJm=0$ and removes the double-excitation component. The tangent is consequently a product-ray tangent, not just the correct marginal velocity. Adding a multiple of $S_2$ shifts only the phase. Minimizing $|u+\lambda|+|v+\lambda|+|\lambda|$ by the median rule gives per-copy cost
\begin{equation}
 P_A=|c|+\frac y{2x}+\frac{1-\min(y^2,z^2)}{4xz}.
\end{equation}
Since $\partial_yC_1=1/q$ at fixed $b$, every strict $P_A<1/q$ is a two-copy saving certificate. This is the third pointwise test in Fig.~2 of the Letter.

\subsection{Diagonal pair flow}\label{sec:diagonal-flow}
Let $\psi_\theta=\cos\theta\ket0+e^{i\pi/4}\sin\theta\ket1$ with $0\le\theta\le\pi/4$. Its Bloch vector is $(\rho/\sqrt2,\rho/\sqrt2,z)$, where $\rho=\sin2\theta,z=\cos2\theta$, and $C_1=\sqrt2\theta$. For $h=(X-Y)/\sqrt2$, $\partial_\theta\psi_\theta=ih\psi_\theta$. The Hamiltonian
\begin{equation}
 H_2(\theta)=h\otimes R_\theta+R_\theta\otimes h
 =\rho(XX-YY)+\frac z{\sqrt2}(XZ+ZX-YZ-ZY)
\end{equation}
generates this product tangent and has per-copy cost $\rho+\sqrt2z$. Start at $\theta_0=\arctan(1/\sqrt2)$, preparing two copies independently, and then integrate this control. The saving per copy is at least
\begin{equation}
 g(\theta)=\sqrt2(\theta-\theta_0)
 +\frac{\cos2\theta-\sqrt2\sin2\theta+1}{2}.
\end{equation}
Here $g(\theta_0)=g'(\theta_0)=0$ and $g''(\theta)=2(\sqrt2\sin2\theta-\cos2\theta)>0$ on the interval. Hence every $\theta>\theta_0$, equivalently $Y>2/3$, saves. Branch propagation covers every northern target with that $Y$, including the diagonal seam itself.

\subsection{The rigorous rational band and equatorial arc}\label{sec:equatorial-arc}
Set $b_0=\tfrac12\arctan(20/21)$, $c_0=21/29$, $s_0=20/29$. On $19/80\le y\le1/4$, the point $r_y(b_0)$ has $x>y$, $z>y$ and $p=c_0$. The $X$-tangent test is $F(y)<0$, where
\begin{equation}
 F(y)=1-(1+2s_0)y^2-2c_0^2y+2c_0q(y-s_0).
\end{equation}
At $y=19/80$, $29^2 80^2F=3319399-44058\sqrt{5678}<0$, because
\begin{equation}
 44058^2\,5678-3319399^2=3197891591>0.
\end{equation}
Throughout the interval, $2/3<s_0<7/10$, $7/10<c_0<3/4$, and $q>9/10$. These give $x>3/5>y$ and $z>91/200>y$. Also $F'=-2K$ with
\begin{equation}
 K=(1+2s_0)y+c_0^2-c_0q+\frac{2c_0y(y-s_0)}q
 >\frac{133}{240}-\frac14-\frac{37}{192}=\frac{107}{960}>0.
\end{equation}
Thus the whole slice saves, and every later point $b\ge b_0$ saves by propagation.

For completeness, the next interval $1/4\le y\le1/\sqrt3$ is certified at $b=\pi/8$. There $x=(y+q)/\sqrt2,z=(q-y)/\sqrt2,p=1/\sqrt2$ and the test is
\begin{equation}
 1-\min(y^2,z^2)+2yz<y+q.
\end{equation}
If the minimum is $y^2$, the difference is $F_1=1-(1+\sqrt2)y^2+\sqrt2yq-y-q$. It is negative at $1/4$, and $F_1'=-L+R_1/q<0$ because
\begin{align}
 L&=1+2(1+\sqrt2)y,\qquad R_1=\sqrt2(1-4y^2)+2y>0,\\
 L^2q^2-R_1^2&=-1+4y+(22+8\sqrt2)y^2+8(\sqrt2-1)y^3-(56+16\sqrt2)y^4>0.
\end{align}
The first pair is nonnegative; the quadratic/quartic pair is positive using $y^2\le1/3$, and the cubic term is positive. If the minimum is $z^2$, put $t=y/q\in[0,1]$ and compare $A_1=1/2+(1+\sqrt2)t+(3/2-\sqrt2)t^2$ with $(1+t)\sqrt{1+2t^2}$. The positive square difference is bounded below by
\begin{equation}
 p_0(t)=\frac34-\frac t2-3t^2+\frac72t^3>0.
\end{equation}
For $t\le1/2$, $p_0$ decreases to $3/16$; for $t=1/2+u$, it is at least $(9/4)(u-7/36)^2+59/576$. This proves both cases. For $y>1/\sqrt3$ on the equator, $Q_Y=1/(2y)<1$ until the diagonal $T$ point, which is covered separately. Combining these intervals proves saving on the complete reduced equatorial arc $\arcsin(19/80)\le\phi\le\pi/4$. No assertion is made about its complement.

\section{Entropy production and the fidelity strong converse}\label{sec:entropy}
For a non-diagonal Pauli with flip mask $f\ne0$, pairs $\{x,x+f\}$ evolve with opposite probability currents. If their probabilities are $a,b$, then $|\dot a|\le2\sqrt{ab}$. A diagonal Pauli preserves all probabilities. For $S_\alpha(q)=\sum_xq_x^\alpha$, a pair contributes at most $2\alpha\sqrt{ab}|a^{\alpha-1}-b^{\alpha-1}|$ to $|\dot S_\alpha|$. Divide by $a^\alpha+b^\alpha$ and put $a/b=e^{-2u}$. Scale homogeneity gives
\begin{equation}
 \frac{2\sqrt{ab}|a^{\alpha-1}-b^{\alpha-1}|}{a^\alpha+b^\alpha}
 =\frac{2\sinh[(1-\alpha)|u|]}{\cosh(\alpha u)}.
\end{equation}
For $1/2<\alpha<1$, this ratio has a finite supremum. At a zero amplitude the current contribution is interpreted by continuity; $|a_x(t)|^{2\alpha}$ is absolutely continuous and differentiable at a zero because $2\alpha>1$. Summing pair bounds gives
\begin{equation}
 |\dot S_\alpha|\le(1-\alpha)\kappa_\alpha S_\alpha,
 \qquad |\dot H_\alpha|\le\kappa_\alpha.
\end{equation}
Integrating from a point-mass input and using entropy additivity proves $C_n\ge nH_\alpha(p)/\kappa_\alpha$.

The endpoint $\alpha=1/2$ can also be handled directly: $\|e^{itP}\varphi\|_1\le(|\cos t|+|\sin t|)\|\varphi\|_1\le e^{|t|}\|\varphi\|_1$. Thus $\log\|\varphi\|_1$ is a tensor-additive Pauli-1-Lipschitz lower bound. Its value on a tensor power is $n\log\sum_x|\psi_x|$.

\subsection{Strict improvement over the amplitude rate}
Let $\alpha=1/2+\delta$, $\delta\downarrow0$. At the maximizing $u_\delta$ of the speed ratio, stationarity gives
\begin{equation}
 (\tfrac12-\delta)\coth[(\tfrac12-\delta)u_\delta]
 =(\tfrac12+\delta)\tanh[(\tfrac12+\delta)u_\delta].
\end{equation}
Writing $v_\delta=e^{-u_\delta}$ makes this
\begin{equation}
 v_\delta\cosh\left(2\delta\log\frac1{v_\delta}\right)=\delta(1+v_\delta^2).
\end{equation}
The maximizer diverges as $\delta\to0$, and this identity gives $v_\delta=\delta(1+o(1))$, $u_\delta=\log(1/\delta)+o(1)$. Substitution yields
\begin{equation}
 \kappa_{1/2+\delta}=2-4\delta\log(1/\delta)+O(\delta).
\end{equation}
For a fixed finite distribution, $H_{1/2+\delta}(p)=H_{1/2}(p)+O(\delta)$, whence
\begin{equation}
 \frac{H_{1/2+\delta}(p)}{\kappa_{1/2+\delta}}
 =\frac12H_{1/2}(p)+H_{1/2}(p)\delta\log(1/\delta)+O(\delta).
\end{equation}
For every coherent target $H_{1/2}(p)>0$, so a small positive $\delta$ gives $C_\infty(\psi)>\log\sum_x|\psi_x|$.

\subsection{Strong converse}
Let $q_x=|\varphi_x|^2$ for an output of length $L$, and $r=p^{\otimes n}$. The entropy law gives $S_\alpha(q)\le e^{(1-\alpha)\kappa_\alpha L}$. The phase-deleting triangle inequality and H\"older's inequality, with exponents $2\alpha$ and $2\alpha/(2\alpha-1)$, give
\begin{align}
 \sqrt{F_n}&\le\sum_x\sqrt{q_xr_x},\\
 F_n&\le S_\alpha(q)^{1/\alpha}
 \left(\sum_xr_x^{\beta}\right)^{(2\alpha-1)/\alpha}
 \le\exp\left[-\frac{1-\alpha}{\alpha}
 (nH_\beta(p)-\kappa_\alpha L)\right],\qquad
 \beta=\frac{\alpha}{2\alpha-1}.
 \label{eq:strong}
\end{align}
As $\alpha\uparrow1$, $\kappa_\alpha\to\kappa_1=2u_*/\cosh u_*$ with $u_*\tanh u_*=1$, while $H_\beta(p)\to H(p)$. For a rate strictly below $H(p)/\kappa_1$, fix $\alpha<1$ sufficiently close to one to obtain a positive exponent. At $\alpha=2/3$, $\kappa_{2/3}=\sqrt2$ and
\begin{equation}
 F_n\le\exp[(\sqrt2L-nH_2(p))/2].
\end{equation}

\begin{figure}[ht]
\centering\includegraphics[width=.55\textwidth]{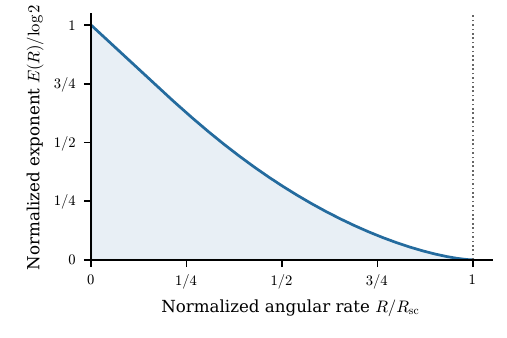}
\caption{Numerical evaluation of the analytic converse exponent for a uniform qubit target, including $T$. The curve evaluates $E(R)=\sup_{1/2<\alpha<1}[(1-\alpha)(\log2-\kappa_\alpha R)/\alpha]_+$ on a finite numerical grid, with axes normalized by $R_{\rm sc}$ and $\log2$. The theorem holds at every $\alpha$ independently of this illustration. The zero at $R_{\rm sc}=\log2/\kappa_1$ is a converse threshold, not an achieved preparation rate.}
\label{fig:converse}
\end{figure}

\section{Exact and approximate growth, including basis targets}\label{sec:growth}
Use trace distance $D(\varphi,\psi)=\tfrac12\|\proj\varphi-\proj\psi\|_1=\sqrt{1-|\langle\varphi|\psi\rangle|^2}$ and define $C_n^\eps$ accordingly. Equation~\eqref{eq:strong} gives
\begin{equation}
 C_n^\eps(\psi)\ge\frac{nH_\beta(p)}{\kappa_\alpha}
 +\frac{\alpha\log(1-\eps^2)}{(1-\alpha)\kappa_\alpha}.
\end{equation}
Therefore any sequence with $\log(1-\eps_n^2)=o(n)$ obeys $\liminf C_n^{\eps_n}/n\ge H(p)/\kappa_1$. A coherent target has $H(p)>0$, while $C_n^{\eps_n}\le nC_1$, so exact, fixed-error, and vanishing-error costs are $\Theta(n)$.

For a nonzero computational label $x$, a single Pauli flip maps the blank to $\ket x^{\otimes n}$ at angle $\pi/2$. Orthogonality and the Fubini--Study bound give
\begin{equation}
 C_n(\ket x)=\pi/2,\qquad C_n^\eps(\ket x)=\arccos\eps\quad(0\le\eps\le1).
\end{equation}
For the approximate upper bound, rotate in the blank/target plane to $\eps\ket{0^{mn}}+\sqrt{1-\eps^2}\ket x^{\otimes n}$, with an irrelevant Pauli phase on the second amplitude. Conversely an output within error $\eps$ of the orthogonal target must lie at Fubini--Study distance at least $\arccos\eps$ from the blank. The blank target itself costs zero. Combined with the entropy lower bound, this proves $C_\infty=0$ iff the target is a computational-basis ray.

If $a=|\langle0^m|\psi\rangle|$, then
\begin{equation}
 C_n^\eps(\psi)=0\quad\Longleftrightarrow\quad
 \eps\ge\sqrt{1-a^{2n}}.
\end{equation}
The reverse direction holds even for an infimum: vanishing cost implies vanishing Fubini--Study distance from the blank. Approximate blocks compose with error
\begin{equation}
 \eps_{n+k}=\sqrt{1-(1-\eps_n^2)(1-\eps_k^2)}.
\end{equation}
Hence fixed-error subadditivity does not follow from exact subadditivity. No fixed-error rate limit or equality with $H(p)/\kappa_1$ is asserted.

\section{Returned catalysts and the attained additive dual}\label{sec:dual}
A returned catalyst is an ancillary pure state which starts uncorrelated with the blank outputs and is restored exactly, uncorrelated with the prepared target. Let $\mathcal A$ be the class of real functions $F$ on all finite-register pure rays, invariant under wire relabeling and uncorrelated blank padding, with
\begin{equation}
 F(0)=0,\qquad F(a\otimes b)=F(a)+F(b),\qquad
 |F(e^{itP}a)-F(a)|\le |t|.
\end{equation}
The middle property is tensor additivity; the last is the Pauli-1-Lipschitz property. Both must hold on correlated as well as product rays. The dual value below depends on the target through evaluation at $\psi$, while the admissible class $\mathcal A$ is common to all targets.
\begin{lemma}[Fixed catalyst elimination]\label{lem:catalyst}
If a finite pure catalyst is returned exactly by a cost-$L$ protocol preparing $k$ target copies, then $C_\infty(\psi)\le L/k$.
\end{lemma}
\begin{proof}
For an $m$-qubit target, pad the catalyst register to $rm$ qubits. Prepare this buffer at finite cost $A$, reuse it on $j$ groups of $k$ blank target blocks, undo its preparation, then prepare $r$ target copies in the buffer at finite cost $B$. This gives an ordinary protocol on $(jk+r)m$ qubits with $C_{jk+r}(\psi)\le A+jL+B$. Divide by $jk+r$ and let $j\to\infty$. The catalyst is fixed before this limit; it is not free workspace in the original model.
\end{proof}

Let $\mathcal M$ be the commutative monoid of finite-register pure rays modulo wire relabeling and uncorrelated blank padding: tensor product is its operation and the blank is its identity. Let $d_0(a,b)$ permit common blank padding returned blank, and define
\begin{equation}
 D(a,b)=\inf_c d_0(a\otimes c,b\otimes c).
\end{equation}
Combined catalysts give the triangle inequality. Also $D(a\otimes c,b\otimes c)=D(a,b)$: tensoring a spectator proves one inequality, and absorbing that spectator into the catalyst proves the other. Form the Grothendieck group of formal tensor-state differences $[a]-[b]$. Its equivalence relation is
\begin{equation}
 (a,b)\sim(a',b')\quad\Longleftrightarrow\quad
 a\otimes b'\otimes c=a'\otimes b\otimes c
 \text{ for some }c\in\mathcal M.
\end{equation}
Translation invariance makes $N([a]-[b])=D(a,b)$ well defined. It is a symmetric subadditive length on the group. Its stable seminorm is
\begin{equation}
 p(g)=\lim_n N(ng)/n=\inf_nN(ng)/n.
\end{equation}
It is integer homogeneous. Lemma~\ref{lem:catalyst} and the empty-catalyst upper bound give $p([\psi]-[0])=C_\infty(\psi)$.

Rationalize the group, quotient the zero seminorm, and complete the resulting rational normed space to a real Banach space. Hahn--Banach supplies a real linear functional of norm at most one attaining the norm of the image of $[\psi]-[0]$. Its pullback $F$ satisfies
\begin{equation}
 F(0)=0,\quad F(a\otimes b)=F(a)+F(b),\quad
 |F(e^{itP}a)-F(a)|\le p([e^{itP}a]-[a])\le|t|.
\end{equation}
Conversely every such $F$ satisfies $nF(\psi)\le C_n(\psi)$ by telescoping. Thus
\begin{equation}
 C_\infty(\psi)=\max_{F\in\mathcal A}F(\psi).
\end{equation}
If $C_\infty=C_1$, an attaining $F$ certifies equality at every block. If any $F$ attains $C_1$, independent preparation and the dual lower bound force $C_n=nC_1$ for all $n$. This existence theorem does not provide an efficiently evaluable optimal certificate.

\subsection{Hereditary saturation and finite-block approximation}
Let $\xi$ precede $\eta$ on a minimizing one-copy path with suffix length $\ell=C_1(\eta)-C_1(\xi)$. Independent suffixes give
\begin{equation}
 C_n(\eta)/n\le C_n(\xi)/n+\ell.
\end{equation}
Hence every defect $\Delta_n=C_1-C_n/n$, and its regularized limit $\Delta=C_1-C_\infty$, is nondecreasing along that path. If $\eta$ has no saving, the same attaining additive $F$ saturates every minimizing prefix $\xi$: $F(\xi)\ge F(\eta)-\ell=C_1(\xi)$, and the reverse inequality always holds. It also certifies additivity of arbitrary heterogeneous products of those prefixes.

On qubit rays, $|C_n(\psi)/n-C_n(\varphi)/n|\le\sqrt3d_{\rm FS}(\psi,\varphi)$ uniformly in $n$. Along factorial blocks $N_k=k!$, $C_{N_k}/N_k$ decreases to $C_\infty$. Continuity and compactness give uniform convergence by Dini's theorem. Every compact subset of the saving set therefore has one common finite-block witness. This supplies no universal witness size at arbitrarily small gap and no evaluated coordinate frontier.

\section{Meridian threshold and stabilizer rates}\label{sec:endpoint}
For $\psi_\theta=\cos\theta\ket0+\sin\theta\ket1$, $0\le\theta\le\pi/2$, $C_1=\theta$. Let $f(\theta)=C_\infty(\psi_\theta)$. Independent angle increments give $|f(\theta)-f(\phi)|\le|\theta-\phi|$. Thus $\theta-f(\theta)$ is continuous, nonnegative, and nondecreasing, with zero set $[0,\theta_c]$. Pauli reflection implies $f(\theta)=f(\pi/2-\theta)$, so $\theta_c\le\pi/4$. The endpoint identity $f(\pi/4)=\pi/4$ is equivalent to $\theta_c=\pi/4$, and then Lipschitzness forces
\begin{equation}
 f(\theta)=\min(\theta,\pi/2-\theta).
\end{equation}
No value of $\theta_c$ is evaluated here. Entropy and independent preparation give
\begin{equation}
 \frac{\log2}{\kappa_1}
 \le c_+:=C_\infty(\ket+)\le\pi/4.
\end{equation}

Every pure stabilizer state $s$ has flat magnitudes on an affine binary support of dimension $r(s)=\log_2|\supp s|$. A Pauli translation removes its affine shift. A blank-fixing binary-linear Clifford brings its support to the first $r$ coordinates, and phase Cliffords remove the linear fourth-root phases and quadratic signs from the stabilizer recurrences. These conjugation symmetries, together with Pauli rate invariance, reduce the state to $\ket+^{\otimes r}\otimes\ket0^{\otimes(m-r)}$. Blank factors do not raise cost, and Lemma~\ref{lem:catalyst} removes any asymptotic benefit from those returned blanks. Consequently
\begin{equation}
 C_\infty(s)=r(s)c_+.
\end{equation}
This is a reduction to one unknown constant, not a determination of all stabilizer costs.

\section{Stabilizer geometry and angle-grid bounds}\label{sec:grid}
For an $M$-qubit ray $v$, let $S_v\subset\mathbb F_2^{2M}$ consist of the Pauli labels for which $v$ is an eigenvector. Eigenvalue signs are omitted. It is isotropic: two anticommuting Paulis cannot share an eigenvector. For binary subspaces define
\begin{equation}
 d_G(A,B)=\dim A+\dim B-2\dim(A\cap B).
\end{equation}
The intersection inequality gives its triangle inequality.

For $U=e^{i\beta P_p}$, the commuting part is unchanged:
\begin{equation}
 S_v\cap p^\perp=S_{Uv}\cap p^\perp.
 \label{eq:common}
\end{equation}
If $\sin2\beta\cos2\beta\ne0$, the old and new spaces cannot both contain a label anticommuting with $p$. To see this, let an old stabilizer be $A$ and a new one be $Q$. Then $U^\dagger Q U$ is a combination of two Paulis $Q$ and $-iP_pQ$ with nonzero coefficients. Exactly one anticommutes with $A$. Projecting its eigenray equation onto the eigenspace of $A$ opposite to $v$ leaves a nonzero Pauli applied to $v$, a contradiction. Equation~\eqref{eq:common} then shows $d_G(S_v,S_{Uv})\le1$.

At $\beta\in(\pi/2)\mathbb Z$ the binary space is unchanged. At $\beta\in\pi/4+(\pi/2)\mathbb Z$, Clifford transvection changes it by distance at most two. For $\beta\in(\pi/8)\mathbb Z$, all cases therefore satisfy
\begin{equation}
 |\beta|\ge\frac\pi8d_G(S_v,S_{Uv}).
\end{equation}
For $\ket0^{\otimes N}\otimes\alpha$ and $\ket+^{\otimes N}\otimes\gamma$, the spaces are respectively $\langle Z_1,\ldots,Z_N\rangle\oplus S_\alpha$ and $\langle X_1,\ldots,X_N\rangle\oplus S_\gamma$. Their intersection has no nonzero output label, so their distance is $2N+d_G(S_\alpha,S_\gamma)\ge2N$. The triangle inequality proves
\begin{equation}
 L\ge N\pi/4.
\end{equation}
The workspace may be arbitrary and changed. At arbitrary small non-Clifford angles, however, losing one exact stabilizer costs arbitrarily little; this distance argument supplies no continuous-angle extensive bound.

\section{Complex population quotients and the positive-product obstruction}\label{sec:matching}
Let $\pi:\mathbb F_2^M\to\mathbb F_2^n$ be a surjective binary quotient. Group computational amplitudes into complex fiber vectors $r_x$, and put $w_x=\|r_x\|_2$. For a skew-Hermitian physical generator $A=i\sum_Pu_PP$, set
\begin{equation}
 K_{xy}=\frac{\operatorname{Re}(r_x^\dagger A_{xy}r_y)}{w_xw_y}
 \quad(x\ne y),\qquad K_{xx}=0,
\end{equation}
with zero assigned at a zero denominator. This real matrix is skew. Differentiating fiber norms shows $\dot w=Kw$ almost everywhere; internal fiber blocks have zero real expectation. A nonnegative absolutely continuous scalar has zero derivative almost everywhere on its zero set, which covers zero fibers.

Write $f(P)$ for the physical $X/Y$ flip mask. A projected nonzero mask $g=\pi f(P)$ gives a perfect matching $\{x,x+g\}$ of quotient labels, with norm-one monomial blocks. Cauchy--Schwarz bounds every induced matching entry by the sum of $|u_P|$ over that mask. Define the enlarged matching charge
\begin{equation}
 \Lambda_n(K)=\sum_{g\ne0}\max_{\{x,x+g\}}|K_{x+g,x}|.
 \label{eq:matchingnorm}
\end{equation}
Then
\begin{equation}
 \Lambda_n(K)\le\sum_{P:\pi f(P)\ne0}|u_P|\le\sum_P|u_P|.
 \label{eq:contraction}
\end{equation}
The enlarged model allows arbitrary matching edge signs; for $n\ge3$ it need not have the same charge as a physical Pauli expansion.

\begin{theorem}[Every-dimensional positive-product tangent bound]\label{thm:product}
For $0\le\theta_i\le\pi/4$, let
\begin{equation}
 r_a=(\cos a,\sin a),\qquad w=\bigotimes_{i=1}^n r_{\theta_i}.
\end{equation}
If $Kw=\sum_iv_i\partial_{\theta_i}w$, then
\begin{equation}
 \Lambda_n(K)\ge\left|\sum_iv_i\right|.
 \label{eq:tangent}
\end{equation}
Consequently any matching path confined to this product cube from a pole to the positive uniform vector costs at least $n\pi/4$.
\end{theorem}
\begin{proof}
First assume positive angles and order $t_i=\tan\theta_i$ so $1\ge t_1\ge\cdots\ge t_n>0$. Put $A_i(0)=-t_i$, $A_i(1)=1/t_i$, and $D_i=\prod_{j>i}(1+t_j^2)$. Define a horizontal cotangent $\lambda_x=w_xR_x$ by
\begin{equation}
 R_x=\sum_i A_i(x_i)D_i\,\mathbf1_{x_{i+1}=\cdots=x_n=0}.
 \label{eq:recursive}
\end{equation}
Under the product probabilities $w_x^2$, $A_i$ has mean zero and mean square one. The suffix factor obeys $D_i\prod_{j>i}\cos^2\theta_j=1$. Independence yields $\lambda\cdot w=0$ and
\begin{equation}
 \lambda\cdot\partial_{\theta_i}w=1\quad\text{for every }i.
\end{equation}
For the prefix ratio $R^{(k)}$, the recursion gives $R^{(k)}_{(p,1)}=1/t_k$ and $R^{(k)}_{(p,0)}=(1+t_k^2)R^{(k-1)}_p-t_k$. Induction using the ordering gives $\max R^{(k)}\le1/t_k$.

The dual slope of a matching $g$ is $H_g=\sum_{\{x,x+g\}}|\lambda_{x+g}w_x-\lambda_xw_{x+g}|$. Let $k$ be the largest set-bit of $g$ and $S=\{i<k:g_i=1\}$. Nonzero suffixes after $k$ have equal ratios and contribute zero. Orient the remaining edges with bit $k$ equal to zero. Their ratio difference is $D_k(1+t_k^2)[1/t_k-R^{(k-1)}_p]\ge0$. Summing amplitude products gives
\begin{align}
 H_g&=Q_S-t_kM_S,\qquad Q_S=\prod_{i\in S}\sin2\theta_i,\\
 M_S&=\sum_{i\in S}\cos2\theta_i
 \prod_{\substack{j<i\\j\in S}}\sin2\theta_j
 \prod_{\substack{i<j<k\\j\in S}}t_j\ge0.
\end{align}
The formula follows by summing each term of Eq.~\eqref{eq:recursive}: the $i$th term vanishes if $i\notin S$, and otherwise its local contribution is $\cos2\theta_i$; the earlier factors give sines and the later zero-bit factors give tangents. Thus $0\le H_g\le Q_S\le1$. It follows that
\begin{equation}
 \left|\sum_iv_i\right|=|\lambda\cdot Kw|
 \le\sum_g\|K_g\|_\infty H_g\le\Lambda_n(K).
\end{equation}
At zero-angle faces, use continuity of the induced tangent norm $\nu(w,d)=\min_{Kw=d}\Lambda_n(K)$. Its minimum exists in finite dimension. Lower semicontinuity follows from compactness of bounded generators; upper semicontinuity follows by correcting a nearby residual $\delta=d'-Kw'$ with the skew matrix $\delta w'^T-w'\delta^T$, whose charge tends to zero. Hence the positive-angle bound extends to all faces.

For an absolutely continuous path in the cube, each angle is absolutely continuous: the norm of the bit-one subvector is $\sin\theta_i$, and $\arcsin$ is Lipschitz on $[0,1/\sqrt2]$. Integrate Eq.~\eqref{eq:tangent} to bound the absolute change of $\sum_i\theta_i$. It is $n\pi/4$ between the stated endpoints.
\end{proof}

Equation~\eqref{eq:contraction} transfers the theorem to any complex physical path whose quotient fiber-norm vector stays product in this cube. The fibers themselves may be entangled and have arbitrary phases. This is not a global scalar calibration on correlated quotient populations.

\subsection{Why a correlated boundary does not tensorize automatically}
An exact obstruction illustrates the last limitation. Let $W=(1,2,2,1)/\sqrt{10}$ and $\delta=(3,0,0,-3)/(4\sqrt{10})$. Its two-bit induced tangent norm is $3/8$: the generator $(3/16)i(IY+YI)$ realizes it, while the cotangent $\lambda=W\odot(9/2,-1/2,-1/2,-1/2)$ gives the same lower bound. At $u=(W,W)/\sqrt2$, let
\begin{equation}
 K=\begin{pmatrix}0&B\\-B^T&0\end{pmatrix},\quad B=S+T,\quad
 S=\begin{pmatrix}0&1/4&1/4&0\\1/4&0&3/4&1/4\\1/4&3/4&0&1/4\\0&1/4&1/4&0\end{pmatrix},
 \quad T_{03}=-T_{30}=3/4.
\end{equation}
Here $SW=W$, $TW=\delta$, so $Ku=(W+\delta,-W+\delta)/\sqrt2$. Its four nonzero matching charges are $1/4,1/4,3/4,0$, totaling $5/4$, below the proposed split $1+3/8=11/8$. Thus a universal logical tangent splitting cannot be inferred from Theorem~\ref{thm:product}. The enlarged generator is not asserted to have physical Pauli charge $5/4$, and this example is not a physical endpoint shortcut.

\section{Two-output and low-deficiency endpoint bounds}\label{sec:deficiency}
This section sharpens the endpoint restrictions with arbitrary-angle results. The fixed two-output theorem is unrestricted in phases, word length, and workspace. The subsequent all-register theorem imposes a bound on the stabilizer deficiency at pulse boundaries.

\subsection{The real two-qubit reduction}
Identify a real four-component vector with the ordinary quaternion $v_{00}+v_{01}\mathbf i+v_{10}\mathbf j+v_{11}\mathbf k$. The six real skew Pauli generators have the following exact actions:
\begin{equation}
\begin{array}{c|rrrrrr}
 P&IY&ZY&YZ&YI&YX&XY\\\hline
 iP&-L_{\mathbf i}&-R_{\mathbf i}&-L_{\mathbf j}&-R_{\mathbf j}&-L_{\mathbf k}&-R_{\mathbf k}
\end{array}
\end{equation}
Left and right multiplications commute, and their total angle charge is the physical charge. A positive product $r_a\otimes r_b$, $r_a=(\cos a,\sin a)$, corresponds to $e^{b\mathbf i}e^{a\mathbf j}$. Lemma~\ref{lem:quaternion} below gives its real preparation cost $a+b$ for $a,b\in[0,\pi/4]$. Both real lifts $\pm r_a\otimes r_b$ are covered. Conjugating a transition by a local Hadamard on every wire whose angle decreases reflects $a\mapsto\pi/4-a$ and preserves the signed odd-$Y$ axis set. Prepending independent preparation then proves
\begin{equation}
 d_{{\rm real},2}(r_{a_1}\otimes r_{a_2},r_{b_1}\otimes r_{b_2})
 =|b_1-a_1|+|b_2-a_2|\quad(0\le a_i,b_i\le\pi/4).
 \label{eq:realprefix}
\end{equation}
The use of a conjugation symmetry does not insert a free Hadamard into the physical circuit.

\subsection{Complex two-bit norm contraction}
For a two-bit quotient, each nonzero projected flip has precisely two matching edges. The two associated odd-$Y$ real axes have edge sign patterns $(+,+)$ and $(+,-)$, up to orientation. Edge coefficients $k_1,k_2$ are represented by $(k_1+k_2)/2$ and $(k_1-k_2)/2$, with exact charge
\begin{equation}
 \left|\frac{k_1+k_2}2\right|+\left|\frac{k_1-k_2}2\right|
 =\max(|k_1|,|k_2|).
\end{equation}
Thus the complex quotient construction of Sec.~\ref{sec:matching} is, for two bits, an exact contraction to real odd-$Y$ Pauli controls. No extra Walsh-sign relaxation is needed. Equation~\eqref{eq:realprefix} therefore bounds every complex physical transition whose two-bit norm vectors are those positive products. In particular,
\begin{equation}
 C_2(\ket+)=\pi/2.
 \label{eq:C2plus}
\end{equation}
More strongly, any two initially blank output bits that finish uniformly populated require charge at least $\pi/2$ on generators flipping at least one of those outputs, with arbitrary initial and final complex workspace and no final independence assumption. This fixed-size identity does not determine $C_\infty(\ket+)$.

\subsection{A compatible potential for codes with at most two logical bits}
Put $q_0=\pi/4$. On nonnegative normalized $k$-bit vectors, $k=1,2$, define
\begin{equation}
 \Phi_k(w)=\min_{x\in\mathbb F_2^k}d_{{\rm real},k}([e_x],[w]),\qquad\Phi_0=0.
\end{equation}
Distance to this finite pole set is 1-Lipschitz. Binary affine relabelings preserve it by conjugating signed real Pauli axes and permuting the poles. The two-bit contraction makes $\Phi_2$ 1-Lipschitz along any complex physical quotient path. The real endpoint bound and completion to uniform give
\begin{equation}
 \Phi_1(r_a)=\bar a,\quad\Phi_2(r_a\otimes r_b)=\bar a+\bar b,
 \quad\bar a=\min(a,\pi/2-a),\qquad
 \Phi_2(e_0\otimes r_a)=\Phi_1(r_a).
 \label{eq:Phiboundary}
\end{equation}
Indeed, from any pole a product target completed to uniform costs at least $2q_0$, and the completion costs $2q_0-\bar a-\bar b$. For the last identity, a pole outside the occupied hyperplane is orthogonal to the target and costs at least $\pi/2$, which cannot beat the nearer in-plane pole.

Let a signed commuting code $T$ have rank $M-k$, $k\le2$. Write $V=\operatorname{flip}(T)$, $\rho=\dim V$, and $D=T\cap Z_M$, so $\dim D=M-k-\rho$. The diagonal constraints give an affine support of dimension $\rho+k$. Its $V$-orbits are the $2^k$ logical computational sectors; stabilizer equations force equal magnitudes within each orbit. Let $w_T(v)$ be the vector of square roots of their probabilities, and define
\begin{equation}
 F_T(v)=\rho q_0+\Phi_k(w_T(v)).
 \label{eq:codepotential}
\end{equation}
Affine invariance makes the definition independent of logical coordinates. If a pulse commutes with $T$, its full path stays in this code. Fixed quotient contraction gives $|F_T(e^{i\beta P}v)-F_T(v)|\le|\beta|$.

If $v$ also has one additional signed stabilizer $s$, then $F_T(v)=F_{T+\langle s\rangle}(v)$. For $\operatorname{flip}(s)\in V$, a representative is diagonal modulo $T$ and pins a logical bit; the last identity in Eq.~\eqref{eq:Phiboundary} applies. Otherwise the extra stabilizer pairs logical sectors with equal magnitudes. Affine coordinates express their norm vector as $r_{q_0}\otimes w'$ and increase the flip rank by one; Eq.~\eqref{eq:Phiboundary} gives the same value. Iterating proves compatibility for any containing code of deficiency at most two. Define $F(v)$ using its full signed Pauli stabilizer group whenever its deficiency is at most two.

\subsection{Generic and quarter-angle pulses}
For a generic angle $\sin2\beta\cos2\beta\ne0$, the argument of Sec.~\ref{sec:grid} says the old and new codes share their commuting part and cannot both have an anticommuting stabilizer. Consequently one boundary code contains the other, and both rays lie in a common signed code of deficiency at most two commuting with $P$. Compatibility and its Lipschitz estimate bound the change in $F$ by $|\beta|$. A Pauli half-turn merely translates logical sectors and preserves $F$.

For a quarter-turn $U=e^{iq_0P_p}$ let $T'=UTU^\dagger$. If the code is preserved and the pulse commutes with it, use the fixed-code estimate. If it is preserved with an anticommuting label, transvection gives $p\in T$, which contradicts isotropy of $T$ with that label; thus this case cannot occur. Otherwise $T_0=T\cap p^\perp$ has rank one less than $T$.

If a diagonal stabilizer is lost, the projected flip $f$ is outside $V$, the diagonal rank decreases by one, and the flip rank increases by one. Each old diagonal logical label $u_i$ can be changed to $u_i+(u_i\cdot f)v$, where $v$ is a lost diagonal label. These new logical observables commute with the pulse and measure the same old probabilities up to translation. Their joint distribution is preserved, giving $F(Uv)-F(v)=q_0$. A gained diagonal stabilizer is the inverse case. These cases cannot occur simultaneously.

If the diagonal rank is unchanged, put $V_0=\operatorname{flip}(T_0)$, of rank $\rho-1$, and choose a lost generator with flip $g$. Then $V=V_0+\langle g\rangle$ and $V'=V_0+\langle g+f\rangle$. If $V'=V$, the unchanged rank forces $f\in V_0$, so all logical diagonal observables commute with the pulse and $F$ is unchanged. Otherwise $f\notin V$. Choose a diagonal label $a$ such that $a\cdot V_0=a\cdot f=0$ and $a\cdot g=1$. Let $\Pi$ select either eigenspace of $Z^a$. Both projected states
\begin{equation}
 \chi=\sqrt2\Pi v,\qquad\chi'=\sqrt2\Pi Uv=U\chi
\end{equation}
are normalized, since the old and new lost stabilizers exchange the two eigenspaces. They lie in the same pinned code $T_0+\langle\pm Z^a\rangle$, of the original rank and flip space $V_0$. Each old or new logical orbit intersects the selected hyperplane in precisely one $V_0$-orbit, and projection followed by the $\sqrt2$ factor preserves its full probability. Thus the old and new logical population vectors are exactly those of $\chi,\chi'$ in this common code, under affine identifications. Its fixed-code estimate bounds the change of the $\Phi_k$ term by $q_0$. The flip rank is unchanged, so $|F(Uv)-F(v)|\le q_0$. Other quarter-angle representatives differ by Pauli half-turns and have at least this charge.

\begin{theorem}[Continuous-angle two-deficiency endpoint bound]
For a finite word on $M$ qubits, suppose every pulse-boundary ray has exact Pauli stabilizer rank at least $M-2$. If the word takes $\ket0^{\otimes M}$ to $\ket+^{\otimes N}\otimes\gamma$, then $L\ge N\pi/4$. The same bound holds with an exactly returned pure catalyst of deficiency at most two, subject to the joint boundary hypothesis.
\end{theorem}
\begin{proof}
The preceding pulse estimates telescope. The blank has $F=0$. Product rays have direct-sum exact stabilizer spaces, so the final potential is $Nq_0+F(\gamma)\ge Nq_0$; the boundary hypothesis ensures $F(\gamma)$ is defined. An exactly returned catalyst contributes the same value to both endpoints and cancels. No arbitrary changed initially supplied workspace is covered by this continuous-angle statement.
\end{proof}
Hence any strict clean-input endpoint saver must reach a pulse-boundary ray with deficiency at least three. This is a necessary condition, not a proof of the unrestricted endpoint.

\section{All-word quaternion lemma used in the two-bit bound}\label{sec:quaternion}
\begin{lemma}\label{lem:quaternion}
Let $D(Q)$ be the infimum of signed coordinate-axis quaternion word lengths reaching $Q$ or $-Q$. For orthogonal imaginary coordinate axes $A,B$ and $0\le a,b\le\pi/4$,
\begin{equation}
 D(e^{aA}e^{bB})=a+b.
\end{equation}
\end{lemma}
\begin{proof}
We prove the homogeneous case $Q_\theta=e^{\theta\mathbf i}e^{\theta\mathbf j}$ first. Its coordinates are
\begin{equation}
 Q_\theta=(\cos^2\theta,\cos\theta\sin\theta,\cos\theta\sin\theta,\sin^2\theta).
\end{equation}
The displayed two-pulse word gives $D\le2\theta$. Suppose a shorter control exists. By compactness of the group and the convex closed coefficient $\ell_1$ ball, a minimum-time control exists; weak-star control convergence and uniform trajectory convergence pass the integral equation to the limit. It agrees with the finite-word infimum by Lemma~\ref{lem:control}. Its time satisfies $T<2\theta\le\pi/2$.

The free-time Pontryagin maximum principle has no abnormal extremal because its maximized Hamiltonian would force a zero costate. Normalize the switching costate by $\|p\|_\infty=1$; its Euclidean radius $\rho$ is constant in $[1,\sqrt3]$. In ordinary quaternion convention $\dot p=2u\times p$, with $\|u\|_1\le1$. A unique maximal coordinate forces the signed coordinate bang; ties permit only convex combinations of maximizing signed axes.

For $\rho<\sqrt2$, there is no tie and the control is one axis, which cannot reach $Q_\theta$. For $\sqrt2<\rho<\sqrt3$, no costate coordinate changes sign. After a proper signed axis change its bangs cycle through three coordinate faces. Their full interior durations are $\lambda=(\pi/2-2\arctan\sqrt{\rho^2-2})/2\in(0,\pi/4)$; the first and last are truncated. At a two-way tie the nonzero third coordinate forces immediate departure, so no dwell occurs.

Four consecutive positive cyclic pulses admit an exact shortening. If their interior angles are $b=c=\lambda$ and endpoint angles $a,d>0$, preserving the endpoint modulus gives $\cos2b\cos2c=\cos^22\lambda$. Vary $b=\lambda+\epsilon$ and solve analytically for $c=\lambda-\epsilon-4\epsilon^2/\sin4\lambda+O(\epsilon^3)$, adjusting $a+d$ and $d-a$ by the endpoint phases. For the two oriented cyclic orders, the quadratic changes in the sum of angles are respectively
\begin{equation}
 -\frac{(1-t)(1+t^2)}{t(1+t)}\epsilon^2,
 \qquad-\frac{(1+t)(1+t^2)}{t(1-t)}\epsilon^2,
 \qquad t=\tan\lambda\in(0,1),
\end{equation}
plus $O(\epsilon^3)$, both strictly negative. These follow by multiplying the four quaternions: the relevant complex pairs are $(\cos b\cos c\pm i\sin b\sin c)e^{i(a+d)}$, with the other pair determining $d-a$. All four angles remain positive for small $\epsilon$. Thus a minimizer in this radius range has at most three positive pulses.

At $\rho=\sqrt3$, the costate is a corner and the only stationary maximizing control is the equal three-axis control. It is locally nonminimal: in matrix convention a short segment $e^{iT(X+Y+Z)/3}$ has exact local Euler representation $e^{icZ}e^{ibY}e^{iaX}$ with $a=T/3-T^2/9+O(T^3)$, $b=T/3+T^2/9+O(T^3)$, $c=T/3-T^2/9+O(T^3)$. Its cost is $T-T^2/9+O(T^3)<T$. The identification $iX\mapsto\mathbf i,iY\mapsto\mathbf j,iZ\mapsto-\mathbf k$ preserves the charged axes; no improper quaternion automorphism is used.

\textit{Three-pulse endpoint exclusion.} Adjacent repeated axes merge. The remaining forms are $ABA$ and $ABC$ with signed angles. For $e^{a\mathbf i}e^{b\mathbf j}e^{d\mathbf i}$, the coordinates are $w=\cos b\cos(a+d)$, $r=\cos b\sin(a+d)$, $u=\sin b\cos(a-d)$, $v=\sin b\sin(a-d)$. Under the hypothesized cost $<\pi/2$, these give $|a|+|d|=\max(|a+d|,|a-d|)$. If the repeated-axis coordinate is $\cos\theta\sin\theta$, then $|b|=\theta$ and $|a+d|=\theta$. If it is $\sin^2\theta$, the bound is $\pi/4+\arcsin(\sin2\theta/\sqrt2)\ge2\theta$; the difference decreases to zero at $\pi/4$.

For an $ABC$ word, multiplication gives
\begin{align}
 w+y&=(\cos b+\sin b)\cos(a+d),& x+z&=(\cos b+\sin b)\sin(a+d),\\
 w-y&=(\cos b-\sin b)\cos(a-d),& x-z&=(\cos b-\sin b)\sin(a-d).
\end{align}
For $\theta<\pi/4$, $w$ exceeds each vector-coordinate magnitude, fixing principal branches with $|b|<\pi/4$. Axis and sign assignments reduce the resulting angle sum to $2\theta$ or one of
\begin{align}
 A(2\theta)&=\arctan(\sin2\theta)+\tfrac12\arcsin(\sin^22\theta),\\
 B(\theta)&=\arctan\frac{t(1+t)}{1-t}+\tfrac12\arcsin(\sin2\theta\cos2\theta),\qquad t=\tan\theta.
\end{align}
Both are at least $2\theta$. For $A(u)\ge u$, set $v=\tan(u/2)$ and compare $\sin^2u$ with $\sin[2(u-\arctan\sin u)]$; after extracting the positive factor $\sin u/(1+\sin^2u)$, the numerator is $2v(v-1)^2(2v^3+5v^2+1)/(1+v^2)^3\ge0$. For $B$, if $\theta\ge\pi/8$ its arctangent already exceeds $2\theta$. Otherwise compare $\sin2\theta\cos2\theta$ with $\sin(4\theta-2p)$, $p=\arctan[t(1+t)/(1-t)]$. The difference is
\begin{equation}
 \frac{4t^2(1-t^2)(2t^3+3t^2-2t+1)}{(1+t^2)^2(t^4+2t^3+2t^2-2t+1)}\ge0.
\end{equation}
The last numerator factor is positive since $3t^2-2t+1=3(t-1/3)^2+2/3$. At $\theta=\pi/4$, the $ABA$ identities force $|b|=|a+d|=\pi/4$; in $ABC$, the middle sign gives either $|b|=\pi/4,|a+d|=\pi/4$, or $b=0,|a+d|=\pi/2$. Thus no three-pulse shortcut exists. The negative lift is excluded throughout by round distance $\arccos(-\cos^2\theta)\ge2\pi/3>\pi/2$.

\textit{The radius $\rho=\sqrt2$.} Switching vertices are proper signed permutations of $(1,1,0)$, with complete coordinate connectors of length $q_0=\pi/4$. At $(1,1,0)$ the sole stationary tied control is $(\mathbf i+\mathbf j)/2$, a dwell $S(d)=e^{d(\mathbf i+\mathbf j)/2}$. Since $T<\pi/2$, at most one complete connector occurs. Transporting a second dwell through its connector combines the two dwell times exactly. A proper signed permutation therefore reduces every path to one of
\begin{align}
 &e^{\beta\mathbf i}S(d)e^{\alpha\mathbf j},\quad
 e^{\beta\mathbf j}S(d)e^{\alpha\mathbf j},\quad
 e^{(\beta+q_0)\mathbf i}S(d)e^{\alpha\mathbf j},\quad
 e^{(\beta+q_0)\mathbf j}S(d)e^{\alpha\mathbf j},\\
 &e^{\beta\mathbf k}e^{q_0\mathbf i}S(d)e^{\alpha\mathbf j},\quad
 e^{-\beta\mathbf k}e^{q_0\mathbf j}S(d)e^{\alpha\mathbf j},
 \qquad 0\le\alpha,\beta\le q_0,\quad d\ge0.
\end{align}
The last four cost an extra $q_0$. This includes both departures from each switching vertex and zero-duration boundary cases.

Here are explicit dominance estimates that exclude all but the first form. Put $C=\cos(d/\sqrt2)$, $s=\sin(d/\sqrt2)/\sqrt2$, so $s/C\le\tan d/2$ when $d<\pi/2$. For the cross-axis form $e^{\beta\mathbf i}S(d)e^{\alpha\mathbf j}$, with $\Sigma=\alpha+\beta$, the vector coordinates are
\begin{equation}
 x=C\cos\alpha\sin\beta+s\cos\Sigma,\quad
 y=C\sin\alpha\cos\beta+s\cos\Sigma,\quad
 z=C\sin\alpha\sin\beta+s\sin\Sigma.
\end{equation}
They are nonnegative. If $\beta\ge\alpha$, then $x\ge y$, with equality only for $\beta=\alpha$, and $x>z$ for $\Sigma>0$ and $\Sigma+d<\pi/2$. For $\Sigma\le\pi/4$ the latter is immediate from $x-z$; otherwise put $e=\pi/2-\Sigma$ and use
\begin{equation}
 \frac{\sin\beta(\cos\alpha-\sin\alpha)}{\cos e-\sin e}
 \ge\frac{\cos e+\sin e-1}{2(\cos e-\sin e)}
 \ge\frac{\tan e}{2}>\frac{\tan d}{2}\ge s/C.
\end{equation}
Interchanging the angles gives the other order. Since the transformed target has two equal largest vector magnitudes, a cross-axis form can match it only with $\alpha=\beta$ and the canonical target orientation. The form with $\beta+q_0$ cannot satisfy that equality below total time $\pi/2$.

For the same-axis form, write $\Delta=\beta-\alpha$. Its coordinates are $x=s\cos\Delta$, $y=C\sin\Sigma+s\cos\Sigma$, $z=-s\sin\Delta$. If $\Sigma>0$, then $y-s=\sin\Sigma[C-s\tan(\Sigma/2)]>0$ under the time bound. Thus $y$ is uniquely largest in magnitude, excluding the target; at $\Sigma=0$ its third coordinate vanishes. The same argument absorbs a $q_0$ connector on the same axis.

For $e^{\beta\mathbf k}e^{q_0\mathbf i}S(d)e^{\alpha\mathbf j}$, $\Sigma+d<q_0$. The three coordinates times $\sqrt2$ are
\begin{align}
 X&=C\cos\Sigma+s(\cos\Delta-\sin\Sigma),\\
 Y&=C\sin\Sigma+s(\cos\Sigma+\sin\Delta),\quad
 Z=C\sin\Sigma+s(\cos\Sigma-\sin\Delta).
\end{align}
They are nonnegative and $X-\max(Y,Z)\ge C(\cos\Sigma-\sin\Sigma)-2s\sin\Sigma>0$, using $\cot\Sigma-1>2\tan d/(1-\tan d)\ge2\tan d\ge2s/C$. This excludes that form. For the final form, before its $-\mathbf k$ pulse set $u=\alpha+q_0$. Its coordinates are $w=C\cos u-s\sin u>0$, $x=s\sin u$, $y=C\sin u+s\cos u$, $z=-s\cos u$. After that pulse $X\ge0,Z\le0$, while $Y-|Z|=C\sin(u-\beta)>0$ and $Y>X$ since $x/y\le s/C\le\tan d/2<\tan(q_0-\beta)$. It too has a unique maximal coordinate.

In the surviving first form, $\alpha=\beta$. Its endpoint costates are $p_0=(\cos2\alpha,1,\sin2\alpha)$ and $p_T=(1,\cos2\alpha,\sin2\alpha)$. The identity $p_T=\operatorname{Ad}(Q_\theta)p_0$ gives $\cos[2(\theta-\alpha)]=1$, hence $\alpha=\theta$ in the allowed range. Its time is $2\theta+d\ge2\theta$, the final contradiction. This completes the homogeneous all-word proof.

For unequal parameters, if $a\ge b$ right-multiply $e^{aA}e^{bB}$ by $e^{(a-b)B}$ to obtain $e^{aA}e^{aB}$. The triangle inequality gives $2a\le D(e^{aA}e^{bB})+a-b$, hence $D\ge a+b$. If $b\ge a$, use the analogous left multiplication. Proper signed coordinate permutations transport the homogeneous statement to any ordered orthogonal signed axes. The upper bound is the two displayed pulses.
\end{proof}

\section{Branch interference and the scope of collision arguments}\label{sec:branches}
For a finite word with pulses indexed chronologically, expand each factor as $\cos\phi_j I+i\sin\phi_jP_j$. A branch $b\in\mathbb F_2^K$ has magnitude $\prod_j|\cos\phi_j|^{1-b_j}|\sin\phi_j|^{b_j}$ and computational endpoint $Fb$, where the columns of $F$ are the flip masks. The exact output amplitude is
\begin{equation}
 a_x=\sum_{b:Fb=x}\zeta(b)\prod_j(\cos\phi_j)^{1-b_j}(\sin\phi_j)^{b_j},
 \qquad |\zeta(b)|=1.
\end{equation}
The Pauli multiplication phases, including the factors of $i$, give $\zeta(b)\in\{1,i,-1,-i\}$. Their phase exponent is quadratic modulo four, and its parity is linear. Thus the branches split into real and imaginary quadratures; terms in opposite quadratures have no probability cross term. Terms in the same quadrature can interfere constructively or destructively in a common fiber.

If $F$ is injective, probabilities are precisely the classical pushforward of independent branch bits with biases $\sin^2\phi_j$. More generally, pairwise quadrature separation in every fiber again removes all cross terms. These are exact algebraic statements about the displayed word. They do not prove a general all-register small-angle collision lower bound. In particular, the number of collision directions alone is not used here to claim additivity on an entire one-Pauli interval.

\section{Figure evaluation and verification conventions}\label{sec:repro}
Figure~\ref{fig:shortcut} of the Letter evaluates the four exact Pauli rotations and checks normalization, endpoint fidelity, and concurrence at the pulse boundaries. Figure~\ref{fig:regions} evaluates the three strict inequalities of Sec.~\ref{sec:regions}, propagates their certified sets along the known minimizing branches, and adds the rigorously proved rational and diagonal bands. Sampling illustrates the sufficient certificates; all analytic bands have separate proofs above. Figure~\ref{fig:converse} evaluates the analytic converse exponent numerically and performs no synthesis optimization. The manuscript uses exact angles, radicals, and rational thresholds for its stated results.

\end{document}

%% file: pretex.tex
\usepackage[dvipsnames]{xcolor}
\usepackage{amsmath,amssymb,amsfonts,mathtools,bm}
\usepackage{graphicx}
\usepackage{booktabs}
\usepackage{enumitem}
\usepackage{hyperref}
\usepackage{url}
\usepackage{cleveref}
\hypersetup{colorlinks=true,citecolor=blue,linkcolor=blue,urlcolor=blue,breaklinks=true}
\newtheorem{theorem}{Theorem}
\newtheorem{proposition}[theorem]{Proposition}
\newtheorem{lemma}[theorem]{Lemma}

\newcommand{\ket}[1]{|#1\rangle}

\newcommand{\proj}[1]{|#1\rangle\!\langle#1|}

\newcommand{\supp}{\operatorname{supp}}

\newcommand{\calP}{\mathcal P}

\newcommand{\eps}{\varepsilon}